\documentclass[11pt,article, reqno]{amsart}
\usepackage{amsmath,amscd,tikz-cd, mathrsfs, amssymb, tikz}
\usepackage{amsmath,amssymb,amsthm, xcolor}
\usepackage[margin=1in]{geometry} 
\usepackage{float, ulem}
\usepackage{graphicx}
\usepackage{color, xcolor}
\usepackage{slashed}
\usepackage{braket}
\usepackage{subcaption}
\usepackage{soul}
\usetikzlibrary{graphs,graphs.standard,quotes}

\usepackage{hyperref}
\usepackage{amsrefs}

\definecolor{ivan}{rgb}{0.9,0.4,1}

\definecolor{ethan}{rgb}{0,.5,1}

\definecolor{thomas}{rgb}{.996, .627, .396}

\definecolor{beata}{rgb}{.8,.1,.5}

\definecolor{sabrina}{rgb}{0.6,0.5,1}

\definecolor{yourname}{rgb}{0.2,0.8,1}

\usepackage{hyperref}

\theoremstyle{plain}
   \newtheorem{theorem}{Theorem}[section]
   \newtheorem{proposition}[theorem]{Proposition}
   
   \newtheorem{corollary}[theorem]{Corollary}

\theoremstyle{definition}
    \newtheorem{definition}[theorem]{Definition}

    \newtheorem{lemma}[theorem]{Lemma}
\newtheorem{example}[theorem]{Example}

\definecolor{green}{RGB}{34, 139, 34}

\theoremstyle{remark}
\newtheorem{remark}{Remark}

\newcommand{\ZZ}{\mathbb{Z}}

\newcommand{\RR}{\mathbb{R}}
\newcommand{\CC}{\mathbb{C}}

\renewcommand{\bar}{\overline}

\renewcommand{\phi}{\varphi}

\makeatletter
\def\bign#1{\mathclose{\hbox{$\left#1\vbox to8.5\p@{}\right.\n@space$}}\mathopen{}}

\begin{document}
\title{Laplace and Dirac Operators on Graphs}
\author{Beata Casiday}
\address{Department of Mathematics\\
Yale University\\
Yale University\\
442 Dunham Lab\\
10 Hillhouse Ave\\
New Haven, CT 06511}
\email{beata.casiday@yale.edu}
\author{Ivan Contreras} 
\author{Thomas Meyer}
\address{Department of Mathematics and Statistics\\
Amherst College\\
31 Quadrangle Drive\\
Amherst, MA 01002}
\email{icontreraspalacios@amherst.edu}
\email{tmeyer23@amherst.edu }
\email{espingarn23@amherst.edu }
\author{Sabrina Mi} 
\address{Department of Mathematics\\
University of Chicago\\
Eckhart Hall\\
5734 S University Ave\\
Chicago IL, 60637}
\email{scmi@uchicago.edu}
\author{Ethan Spingarn}
\maketitle
\begin{abstract}
Discrete versions of the Laplace and Dirac operators haven been studied in the context of combinatorial models of statistical mechanics and quantum field theory. In this paper we introduce several variations of the Laplace and Dirac operators on graphs, and we investigate graph-theoretic versions of the Schr\"odinger and Dirac equation. We provide a combinatorial interpretation for solutions of the equations and we prove gluing identities for the Dirac operator on lattice graphs, as well as for graph Clifford algebras. \end{abstract}
\tableofcontents

\section{Introduction}
The Laplace operator, or Laplacian, is a fundamental object of study in mathematics and physics. In particular, the evolution of quantum-mechanical systems is controlled by the Schr\"odinger equation, which relies on the properties of the Laplace operator. Discretized versions of the Laplacian have been implemented in order to study combinatorial models in quantum field theory \cite{Mnev16, Reshetikhin:2014jaa, Cimasoni07}. Similarly, the Dirac operator (that can be understood as a \textit{square root} of the Laplacian) is part of the mathematical formulation of spinors. 
One of the purposes of this paper is to study different versions of the Laplace and Dirac operators on finite graphs, and to analyze the (time dependent) Dirac equation on graphs, which gives a graph-theoretic interpretation of spinors. The notion of spinor has its origins in particle physics: it first appeared around 1922 when Stern and Gerlach realized that electrons can be catalogued into two groups or streams (``up" or ``down"), depending on a separation by a non-uniform magnetic field.
It was later interpreted as a version of angular momentum, so the spin of a particle is naturally associated to the notion of rotation.
Mathematically speaking, we can think of spinors as  vectors in $\mathbb C^2$. For instance, $\begin{bmatrix}1\\0 \end{bmatrix}$ can be seen as a spinor. Spinors have natural linear transformations, which can be seen as rotations in $\mathbb R^3$. More generally, it turns out that quaternions are useful to describe higher dimensional rotations, and therefore, spinors.

From the point of view of quantum mechanics, the dynamics of a quantum particle can be described in terms of the Schr\"odinger equation:
\begin{equation}
\partial_t(\psi)= \frac{i}{\hbar} \Delta (\psi), 
\end{equation}
where $\psi$ is a state (a special type of function) and $\Delta$ is the Laplace operator:
\[\Delta = \nabla \cdot \nabla,\]
where $\nabla$ denotes the gradient.
Dirac provided a way to include particles with spin in the equation:
\begin{equation}
\slashed {\partial} (\psi)= \frac{mc}{i\hbar} \psi 
\end{equation}
where $\slashed{\partial}$ is called the Dirac operator, satisfying $\slashed{\partial}^2= \Delta$.

A toy model of these equations is given in terms of the discrete Laplace and Dirac operators. The key advantage is that the evolution of states in this toy model depends only on the spectral properties of the graph. In particular, spectral graph theory analyzes the features of a graph in relationship with the behavior of the (eigenvalues/eigenvectors of) matrices associated with that graph.

Given the adjacency matrix $A$, and the degree matrix $D$, for a given graph $\Gamma$, the Laplacian $\Delta(\Gamma)$ of the graph is defined as the following matrix:
       \begin{equation*}
       \Delta(\Gamma) = D(\Gamma) - A(\Gamma).
       \end{equation*}
In Section 2 we introduce the even and odd graph Laplacian matrix, and the main results (Theorems \ref{thm: Even_Schrodinger} and \ref{thm:Odd_Schrodinger}) describe the steady states for the even and odd versions of the graph Schr\"odinger equations, in terms of the connected components and inependent cycles of the graph.

In Section 3 we introduce various versions of the Dirac operator on graphs, including the \textit{incidence Dirac operator} (Definition \ref{def: Incidence_Dirac}), inspired by the work of Knill \cite{Knill}. The main result of this section (Theorem \ref{thm: Incidence_Dirac}) is a graph-theoretic interpretation of the powers of the incidence Dirac operator, in terms of the number of certain walks on the graph.
A particular type of Dirac operators on graphs appears in the context of dimer models, following the work of Kenyon \cite{Kenyon2002}, Cimasoni and Reshetikhin \cite{Cimasoni07}. It turns out that the Kasteleyn matrix, which produces the number of perfect matchings of lattice graphs, can be interpreted as a discrete Dirac operator. We prove (see Theorems \ref{Theo:G2,2}, \ref{Theo:T9}, \ref{Theo:T10}, \ref{Theo:3B3}) gluing formulae for different cases of graph gluing of lattice graphs. 
On the other hand, spinors have an algebraic representation via Clifford algebras. In Section 5 we follow the construction of Clifford algebras for graphs introduced by Khovanova in \cite{Khovanova}. There it is described how to assign Clifford algebras to graphs. The main results (Theorems \ref{thm: Clifford_Gluing_Center} and \ref{thm: Clifford_Gluing_Path}) provide an algebraic interpretation for the gluing of graphs in terms of their corresponding Clifford algebras.

\subsection{Acknowledgments} This research project started during the  2021 SUMRY Program at Yale University, which was supported by NSF (DMS-2050398). I.C. thanks Pavel Mnev for useful discussions during the early stages of this project.
\section{The graph Laplacian and the Schr\"odinger equation}\label{sec:Laplace_Schrod}

The Laplacian describes the evolution of a quantum state over time for particles without spin, as governed by the Schr\"odinger equation. We study graph theoretic analogues of the Laplacian \cite{nica_SpectralGraph} 
and its interpretation in quantum mechanics \cite{Mnev16}.
We provide characterizations of the steady states of the graph Schr\"odinger equation and give a result on the average of quantum states.

We begin by defining quantum states on a graph. In the continuum, a quantum state is an assignment of a complex number to each point in space. We consider a graph theoretical model in which a graph quantum state is an assignment of a complex number to each vertex and each edge of a finite simple graph. This model can be considered as taking a "sampling" of points from the continuum and restricting our study to the evolution of this sample over time.

We next define several operators that are integral to our study of graph quantum mechanics.

\begin{definition}
Let $\Gamma = (V, E)$ be a finite simple graph. The \textit{incidence matrix} of $\Gamma$, denoted $I$, is the $|V| \times |E|$ matrix whose $(i,j)$ entry is defined by
\[[I]_{i, j} = \begin{cases}
      1 & \text{edge $j$ ends at vertex $i$} \\
     -1 & \text{edge $j$ starts at vertex $i$} \\
      0 & \text{otherwise.} 
   \end{cases}
\]
\end{definition}

We use the incidence matrix to define the even and odd graph Laplacians. These operators act as discretized forms of the Laplace operator, and act on the vertices and edges of our graph respectively.

\begin{definition} \label{def:Quantum State}
Let $\Gamma = (V, E)$ be a finite simple graph. A \textit{vertex state} assigns a complex number to each vertex $v \in V$. An \textit{edge state} assigns a complex number to each edge $e \in E$. A \textit{vertex-edge state} assigns a complex number to each $v \in V$ and each $e \in E$.
\end{definition}

\begin{remark}
When clear from context, we refer to these states as quantum states. In general, the even and odd Laplacian and Dirac operators act on vertex states and edge states respectively, and the incidence Dirac operator acts on vertex-edge states. We use the notation $\psi$ to refer to a general quantum state, and use $v$ and $e$ to denote vertex and edge states respectively.
\end{remark}

\begin{definition} \label{def:Even Laplacian}
Let $I$ be the incidence matrix of a finite simple graph $\Gamma$. The \textit{Even Graph Laplacian} is defined as
\[ \Delta_+ = I I^t.\]
Equivalently, $\Delta_+$ can be defined as \[ \Delta_+ = D - A, \] where $D$ and $A$ are the degree and adjacency matrices of $\Gamma$ respectively.
\end{definition}

\begin{definition}\label{def: Odd Laplacian} 
The \textit{Odd Graph Laplacian} is defined as 
\[ \Delta_- = I^t I. \]
\end{definition}

We now give a concrete example of the even and odd Laplacians on a small graph.

\begin{example}
Let $\Gamma$ be the path graph with 3 vertices, and orientation given below.

\begin{figure}[h]
\centering
\includegraphics[scale=0.6]{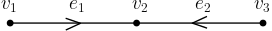}
\caption{Oriented path graph $P_3$.}
\end{figure}

Then the previously defined operators are given by:
\[ I = 
\begin{bmatrix}
-1 & 0 \\
1 & 1 \\
0 & -1
\end{bmatrix} \]
\[ \Delta_+ = II^t = \begin{bmatrix}
-1 & 0 \\
1 & 1 \\
0 & -1
\end{bmatrix} 
\begin{bmatrix}
-1 & 1 & 0 \\
0 & 1 & -1
\end{bmatrix} 
= \begin{bmatrix}
1 & -1 & 0 \\
-1 & 2 & -1 \\
0 & -1 & 1
\end{bmatrix} \]
\[ \Delta_- = I^t I = 
\begin{bmatrix}
-1 & 1 & 0 \\
0 & 1 & -1
\end{bmatrix} 
\begin{bmatrix}
-1 & 0 \\
1 & 1 \\
0 & -1
\end{bmatrix} 
= \begin{bmatrix}
2 & 1 \\
1 & 2
\end{bmatrix}. \]
\end{example}

\begin{remark}
By definition, we have $\Delta_+ = D - A$, so that the even Laplacian is independent of the orientation of $\Gamma$. In contrast, $\Delta_-$ does depend on orientation of the underlying graph, as shown in the following example.
\end{remark}

\begin{example}\label{Ex:OrientedC3}
Let $\Gamma_1$ and $\Gamma_2$ be $C_3$ with the orientations given below:

\begin{figure}[h]
\includegraphics[scale=0.4]{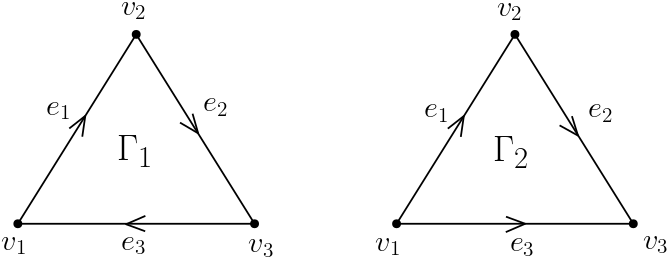}
\caption{Two orientations of $C_3$.}
\end{figure}


A routine calculation yields \\
$\Delta_- (\Gamma_1) = \begin{bmatrix}
2 & -1 & -1 \\
-1 & 2 & -1 \\
-1 & -1 & 2 \\
\end{bmatrix}$

$\Delta_- (\Gamma_2) = \begin{bmatrix}
2 & -1 & 1 \\
-1 & 2 & 1 \\
1 & 1 & 2 \\
\end{bmatrix}$.
\end{example}

\begin{remark}
Eigenvalues of the odd Laplacian are independent of orientation. Indeed, for two graphs $\Gamma_1$ and $\Gamma_2$ of differing orientations, the odd Laplacians $\Delta_- (\Gamma_1)$ and $\Delta_- (\Gamma_2)$ are related by conjugation, hence have the same eigenvalues.
\end{remark}


With these operators in hand, we introduce the discrete Schr\"odinger equation.

\begin{definition}
Let $\Gamma$ be a finite simple graph with m vertices and n edges, and let $\psi: \RR \to \CC^m$ ($\CC^n$) be a function assigning a complex number to each vertex (edge) of $\Gamma$. The \textit{discrete Schr\"odinger equation} is defined as
\[
\partial_t\psi = \frac{i}{\hbar}\Delta_{\pm}\psi .\]
\end{definition}

\begin{remark}

\end{remark}

The following theorem \cite{Mnev16} characterizes the solutions.

\begin{theorem}\label{thm: Mnev16}

Solutions to the discrete Schr\"odinger equation are given by
\[
        \psi(t)= e^{\frac{i}{\hbar} \Delta_{\pm}t}\psi(0).
\] 
\end{theorem}
From now on, we will write $\psi(0)$ as $\psi_0$ or simply $\psi$.
\subsection{Steady states}

We now proceed to characterize steady states of the even and odd graph Laplacians. These states are constant under the time evolution of the graph Schr\"odinger equation. From these steady states, we can extract graph-theoretic information on the connected components and independent cycles of our graph.

We show that steady states are in bijection with the kernels of the even and odd graph Laplacians, and identify these kernels accordingly.

\begin{definition}
A quantum state $\psi_0$ is a \textit{steady state} if $e^{\frac{i}{\hbar}\Delta_{\pm} t} \psi_0 = \psi_0$ for all $t \geq 0$.
\end{definition}

\begin{theorem}
\label{Theo: steady states in bijection with kernel}

Let $A$ be an $n \times n$ matrix, $\psi_0 \in \mathbb{C}^n$ be a quantum state, $k\in \mathbb{C}$ be a nonzero complex constant, and $t$ a complex variable. Then $e^{kAt}\psi_0=\psi_0$ for all values of $t$ if and only if $\psi_0 \in \ker(A)$.
\end{theorem}
\begin{proof}
$(\Leftarrow)$ Suppose $\psi_0 \in \ker(A)$, so $A\psi_0=0$. Then \begin{align*}
    e^{kAt}\psi_0 &=(\sum_{n=0}^\infty \frac{(kAt)^n}{n!})\psi_0\\
    &=(I+kAt+\frac{k^2A^2t^2}{2}+\frac{k^3A^3t^3}{3!}+...)\psi_0\\
    &=I\psi_0+kAt\psi_0+\frac{k^2A^2t^2}{2}\psi_0+\frac{k^3A^3t^3}{3!}\psi_0+\cdots\\
    &=I\psi_0+kA\psi_0t+\frac{k^2A^2\psi_0t^2}{2}+\frac{k^3A^3\psi_0t^3}{3!}+\cdots \\
    &=I\psi_0+k(0)t+\frac{k^2A(0)t^2}{2}+\frac{k^3A^2(0)t^3}{3!}+\cdots \\
    &=I\psi_0+0+0+0\cdots= I\psi_0=\psi_0.\\
\end{align*}
$(\Rightarrow)$ Suppose $e^{kAt}\psi_0=\psi_0$ for all values of $t$. Then $\frac{\partial}{\partial t}e^{kAt}\psi_0=\frac{\partial}{\partial t}\psi_0$. The right-hand side of this equation equals $0$ as $\psi_0$ is constant for all $t$. Now, considering the left hand side of this equation, we obtain:
\begin{align*}
    \frac{\partial}{\partial t}e^{kAt}\psi_0&=\frac{\partial}{\partial t}(\sum_{n=0}^\infty \frac{(kAt)^n}{n!})\psi_0\\
    &=\frac{\partial}{\partial t}(I+kAt+\frac{k^2A^2t^2}{2}+\frac{k^3A^3t^3}{3!}+\cdots)\psi_0\\
    &=(\frac{\partial}{\partial t}I+\frac{\partial}{\partial t}kAt+\frac{\partial}{\partial t}\frac{k^2A^2}{2}t^2+\frac{\partial}{\partial t}\frac{k^3A^3}{3!}t^3+\cdots)\psi_0\\
    &=(0+kA+k^2A^2t+\frac{k^3A^3}{2}t^2+\cdots)\psi_0\\
    &=kA(I+kAt+\frac{k^2A^2t^2}{2}+\frac{k^3A^3t^3}{3!}+\cdots)\psi_0\\
    &=kAe^{kAt}\psi_0.
\end{align*}
By assumption, $kAe^{kAt}\psi_0=kA\psi_0$. This means $kA\psi_0=0$, so $\psi_0 \in \ker(A)$.
\end{proof}

Since we have proven that steady states are in bijection with elements in the kernel of the graph Laplacians, it suffices to find the kernels of the even and odd graph Laplacians in order to characterize our steady states.


We now characterize the kernel of the even Laplacian. We begin with an intermediate lemma.

\begin{lemma}
Let $\Gamma$ be a graph with connected components $\{1, 2, \dots, b_0\}$. A quantum state $\psi$ is in $\ker \Delta_+$ if $\psi = 1$ on the vertices of the $i$th connected component, and $0$ on all other vertices, for some $i \in \{ 1, \dots, b_0 \}$.
\end{lemma}

We digress to provide a brief example, from which the general method of proof is clear.


\begin{example}
Let $\Gamma$ be the oriented graph below, with two cycles and two connected components.

\begin{figure}[h]
\includegraphics[scale=0.4]{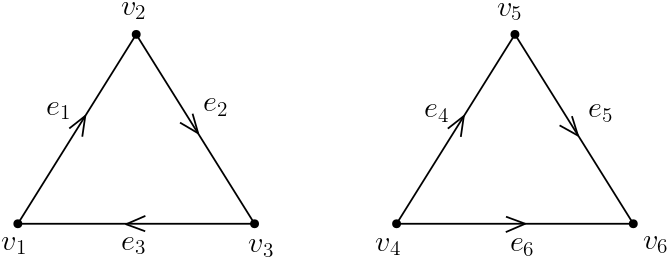}
\caption{Oriented graph with two cycles and two connected components.}
\end{figure}

The even Laplacian of $\Gamma$ is
\[ \Delta_+ = \begin{bmatrix}
2 & -1 & -1 & 0 & 0 & 0 \\
-1 & 2 & -1 & 0 & 0 & 0 \\
-1 & -1 & 2 & 0 & 0 & 0 \\
0 & 0 & 0 & 2 & -1 & -1 \\
0 & 0 & 0 & -1 & 2 & -1 \\
0 & 0 & 0 & -1 & -1 & 2 
\end{bmatrix}. \]
Let $\psi = \begin{bmatrix}
1 \\
1 \\
1 \\
0 \\
0 \\
0 
\end{bmatrix}$ be the quantum state with entries $1$ on the elements of the leftmost connected component of the graph, and $0$ on all other vertices. We immediately see that

\[\Delta_+ \psi = \begin{bmatrix}
2 & -1 & -1 & 0 & 0 & 0 \\
-1 & 2 & -1 & 0 & 0 & 0 \\
-1 & -1 & 2 & 0 & 0 & 0 \\
0 & 0 & 0 & 2 & -1 & -1 \\
0 & 0 & 0 & -1 & 2 & -1 \\
0 & 0 & 0 & -1 & -1 & 2 
\end{bmatrix} 
\begin{bmatrix}
1 \\
1 \\
1 \\
0 \\
0 \\
0 
\end{bmatrix} = 0,\]
so that $\psi$ is in the kernel of $\Delta_+$ as was claimed in our lemma.
\end{example}

We now proceed to a formal proof of this lemma.

\begin{proof} 

Let $\Gamma$ be a graph with $n$ vertices. Let $v_1, v_2, \dots, v_k$ be the vertices of a connected component of $\Gamma$.

Let 
\[ \Delta_+ = \begin{bmatrix}
\Delta_{1,1} & \dots & \Delta_{1,n}\\
\vdots & \ddots & \vdots\\
\Delta_{n,1} & \dots & \Delta_{n,n}
\end{bmatrix}
\]

and let
\[ \psi =
\begin{bmatrix}
1 \\
\vdots \\
1 \\
0 \\
\vdots \\
0
\end{bmatrix}
\]
be the column vector with values $[\psi]_{i,1} = \begin{cases}
    1 & 1 \leq i \leq k \\
    0 & k < i \leq n
\end{cases}
.$

Multiplying by the even Laplacian, we observe that
\[ \Delta_+ \psi = 
\begin{bmatrix}
\Delta_{1,1} \cdot 1 + \dots + \Delta_{1,k} \cdot 1 + \Delta_{1,k+1} \cdot 0 + \dots + \Delta_{1,n} \cdot 0 \\
\vdots \\
\Delta_{k,1} \cdot 1 + \dots + \Delta_{k,k} \cdot 1 + \Delta_{k,k+1} \cdot 0 + \dots + \Delta_{k,n} \cdot 0 \\
\Delta_{k+1,1} \cdot 1 + \dots + \Delta_{k+1,k} \cdot 1 + \Delta_{k+1,k+1} \cdot 0 + \dots + \Delta_{k+1,n} \cdot 0 \\
\vdots \\
\Delta_{n,1} \cdot 1 + \dots + \Delta_{n,k} \cdot 1 + \Delta_{n,k+1} \cdot 0 + \dots + \Delta_{n,n} \cdot 0 \\
\end{bmatrix}
=
\begin{bmatrix}
\Delta_{1,1} + \dots + \Delta_{1,k}\\
\vdots \\
\Delta_{k,1} + \dots + \Delta_{k,k}\\
\Delta_{k+1,1} \dots + \Delta_{k+1,k}\\
\vdots \\
\Delta_{n,1} + \dots + \Delta_{n,k}\\
\end{bmatrix}.
\]
Recall that $\Delta_+ = D - A$, where $D$ and $A$ denote the degree and adjacency matrices of $\Gamma$ respectively. Thus, for all $\Delta_{i,j}$ with $i > j$ and $j \in \{1, 2, \dots, k \}$, we have $\Delta_{i,j} = 0$, as vertices $i$ and $j$ are in different connected components of $\Gamma$. Thus, we have
\[
\Delta_+ \psi =
\begin{bmatrix}
\Delta_{1,1} + \dots + \Delta_{1,k}\\
\vdots \\
\Delta_{k,1} + \dots + \Delta_{k,k}\\
\Delta_{k+1,1} \dots + \Delta_{k+1,k}\\
\vdots \\
\Delta_{n,1} + \dots + \Delta_{n,k}\\
\end{bmatrix}
=
\begin{bmatrix}
\Delta_{1,1} + \dots + \Delta_{1,k}\\
\vdots \\
\Delta_{k,1} + \dots + \Delta_{k,k}\\
0\\
\vdots \\
0\\
\end{bmatrix}.
\]

Note that both $D(v_i)$ and $\sum_{j=1}^k A_{i,j}$ count the number of edges incident to the vertex $v_i$, hence are equal. Thus, we obtain our desired result, as by definition we have $\Delta_{i,1} + \dots + \Delta_{i,k} = D(v_i) - \sum_{j=1}^k A_{i,j} = 0$, so that $\Delta_+ \psi = 0$ as desired.
\end{proof}

\begin{lemma} \label{Theo: kernel of even laplacian iff constant on connected comp}
A quantum state is in the kernel of the even Laplacian if and only if it is constant on each of the connected components of the graph.
\begin{proof}

From \cite{Contreras19}, we know that $\dim (\ker \Delta_+) = b_0$. Label our connected components $\{ 1, 2, \dots, b_0 \}$. Let $\psi_i$ be the vertex state with the value $1$ on the vertices of the $i$th connected component, and $0$ on all other vertices. By the above lemma, each of these $\psi_i$ is in the kernel of $\Delta_+$. Furthermore, each $\psi_i$ is linearly independent, and we have one $\psi_i$ for each connected component, so these $\psi_i$ must span $\ker \Delta_+$ as desired.
\end{proof}
\end{lemma}

\begin{corollary}
\label{Even Schrodinger steady states}
A vertex state $\psi$ is steady if and only if $\psi$ is constant on each of the graph's connected components.
\begin{proof}
Combining results \ref{Theo: steady states in bijection with kernel} and \ref{Theo: kernel of even laplacian iff constant on connected comp}, we see that a vertex state is steady if and only if it is in the kernel of the even Laplacian if and only if it is constant on a graph's connected components.
\end{proof}
\end{corollary}

\begin{corollary} The vector space of steady vertex states is isomorphic to $\mathbb C^{b_0}$.
\end{corollary}

\subsection{Average of quantum states}

We now proceed to characterize quantum states (both vertex and edge states) which have a constant average under the Schr\"odinger equation.

\begin{definition}
Let $\psi = \begin{bmatrix}
z_1 \\
z_2 \\
\vdots \\
z_n
\end{bmatrix}$
be a vertex state of a graph. We define the \textit{average of $\psi$} to be
\[ \mu(\psi) = \frac{z_1 + z_2 + \dots + z_n}{n}.\]
\end{definition}

\begin{definition} Let $\Delta_{\pm}$ denote the even and odd graph Laplacians. The \textit{partition function} of $\Delta_{\pm}$ is defined to be
\[ Z(t) = e^{\frac{i}{\hbar}\Delta_{\pm}t}. \]
\end{definition}

\begin{remark}
The even Laplacian is a real symmetric matrix so it is Hermitian.
\end{remark}

\begin{lemma}
The partition function of a Hermitian matrix is unitary. In particular, the partition function of the even Laplacian is unitary.
\begin{proof}
Let $A$ be a Hermitian matrix; that is, $A = \bar{A}^T$. We want to show $e^{iA}=I+iA+\frac{(iA)^2}{2}+\frac{(iA)^3}{3!}+...$ is unitary; that is, $\bar{e^{iA}}^T=e^{-iA}$. There are a few linear algebra facts that will be useful in this proof:
\begin{enumerate}
    \item $\bar{A+B}^T=\bar{A}^T+\bar{B}^T$
    \item $\bar{(A^n)}^T=(\bar{A}^T)^n$
    \item $\bar{iA}=\bar{i}\text{ }\bar{A}=-i\bar{A}$.
\end{enumerate}
The proof consists of a string of equalities: 
\begin{align*}
    \bar{e^{iA}}^T&=&\bar{I+iA+\frac{(iA)^2}{2}+\frac{(iA)^3}{3!}+\cdots}^T\\
    &=&I+\bar{iA}^T+\frac{\bar{(iA)^2}^T}{2}+\frac{\bar{(iA)^3}^T}{3!}+\cdots\\
    &=&I+\bar{iA}^T+\frac{(\bar{iA}^T)^2}{2}+\frac{(\bar{iA}^T)^3}{3!}+\cdots\\
    &=&I+-i\bar{A}^T+\frac{(-i\bar{A}^T)^2}{2}+\frac{(-i\bar{A}^T)^3}{3!}+\cdots\\
    &=&I+-iA+\frac{(-iA)^2}{2}+\frac{(-iA)^3}{3!}+...\\
    &=&e^{-iA}.
\end{align*}
\end{proof}
\end{lemma}


\begin{theorem}\label{thm: Even_Schrodinger}
For a vertex state $\psi_0$, the average $\mu(\psi_0)=\mu(\psi_t)$ for all $t \geq 0$.
\end{theorem}
\begin{proof}
Let $\psi_0 = \begin{bmatrix}
z_1 \\
z_2 \\
\vdots \\
z_n
\end{bmatrix}$ be a vertex state on a graph. The average of $\psi_0$ can be given by the inner product $\langle v, \psi_0 \rangle$, where $v = \begin{bmatrix}
\frac{1}{n} \\
\frac{1}{n} \\
\vdots \\
\frac{1}{n}
\end{bmatrix}$. The partition function $Z(t) = e^{\frac{i}{\hbar} \Delta_+ t}$ of the even Laplacian is unitary, hence preserves inner product. From this, we have \[ \mu(\psi_0) = \langle v , \psi_0 \rangle = \langle Z(t) v, Z(t) \psi_0 \rangle = \langle v, \psi_t \rangle = \mu(\psi_t),\] where the equality $Z(t) v = v$ follows as $v$ is a steady state (constant on all vertices), and $Z(t) \psi_0 = \psi_t$ simply denotes a solution to the Schr\"odinger equation over time. Thus the average is constant as desired.
\end{proof}

Having characterized the steady states of the even graph Schr\"odinger equation, we now turn our attention to the odd graph Schr\"odinger equation.

\begin{definition}
Recall that the first Betti number $b_1$ is equivalent to the dimension of the cycle space of a graph. Let $c_1, c_2, \dots, c_{b_1}$ be the independent cycles of a graph $\Gamma$. To each independent cycle $c_i$, we define the \textit{independent cycle edge state} $\alpha_i$ to be the edge state with value $1$ on all edges with clockwise orientation on $c_i$, $-1$ on all edges with counterclockwise orientation on $c_i$, and $0$ on all edges not in $c_i$.
\end{definition}

\begin{theorem}\label{thm:Odd_Schrodinger}
An edge state $\psi$ is steady for the odd Schr\"odinger if and only if $\psi$ represents a linear combination of independent cycle edge states.
\end{theorem}

Note that this theorem is not independent of orientation; that is, two orientations of the same graph may admit two different spaces of steady states of the odd Laplacian.

Let us consider $C_3$ with the same orientations as in Example \ref{Ex:OrientedC3}. The theorem above tells us that the leftmost graph has steady states $\ker(\Delta_-)  = \displaystyle{span \left\{ 
\begin{bmatrix}
1\\
1\\
1
\end{bmatrix} \right\}}$,
while the rightmost graph has steady states $\ker(\Delta_-)  = \displaystyle{span \left\{ 
\begin{bmatrix}
1\\
1\\
-1
\end{bmatrix} \right\}}$.
With this example in hand, we now prove the theorem in general.


\begin{proof}
By \ref{Theo: steady states in bijection with kernel}, it suffices to show that the kernel of $\Delta_-$ is spanned by linear combinations of independent cycles. To this end, we show that every vector representing an independent cycle is in the kernel of the incidence matrix. 

We first consider the action of the incidence matrix on a single edge $e_k$; that is, on the edge state with value $1$ on the edge $e_k$, and value $0$ on all remaining edges. By definition, we have
\[ [I]_{i,j} = 
\begin{cases}
-1 & \text{edge $e_j$ begins at vertex $v_i$} \\
1 & \text{edge $e_j$ ends at vertex $v_i$} \\
0 & \text{otherwise.} \\
\end{cases} \]

Let $v_{e_k,i}$ and $v_{e_k,f}$ denote the initial and final vertices of $e_k$ respectively. By the above definition, using a slight abuse of notation, we see that $I e_k = v_{e_k,f} - v_{e_k,i}$ (where here, $e_k$, $v_{e_k,i}$, and $v_{e_k,i}$ represent the edge and vertex states with value $1$ on $e_k$, $v_{e_k,i}$, and $v_{e_k,i}$ respectively and $0$ on all other edges and vertices).

We now consider the action of $I$ on an independent cycle edge state. Let $\alpha_i$ be an independent cycle edge state. We many decompose $\alpha_i$ into a sum of its individual edges $\alpha_i = e_{i,1} + e_{i,1} + \cdots + e_{i,k_i}$. Without loss of generality, we may assume the final vertex of each edge $e_{i,j}$ is the initial vertex of the edge $e_{i,j+1}$ and the final vertex of $e_{i,k_i}$ is the initial vertex of $e_{i,1}$). By linearity, we obtain $I \alpha_i = I e_{i,1} + I e_{i,2} + \cdots + I e_{i,k_i}$. Combining this with our previous consideration of the action of $I$ on individual edges, we know
\[ I e_{i,1} + I e_{i,2} + \cdots + I e_{i,k_i} = (v_{e_1,f} - v_{e_1,i}) + (v_{e_2,f} - v_{e_2,i}) + \cdots + (v_{e_{k_i},f} - v_{e_{k_i},i}).
\]

Since we have chosen our independent cycle edge states to have value $1$ on clockwise oriented edges and $-1$ on counterclockwise oriented edges, and our edges are assumed to be consecutive, we see that
\[ (v_{e_1,f} - v_{e_1,i}) + (v_{e_2,f} - v_{e_2,i}) + \cdots + (v_{e_{k_i},f} - v_{e_{k_i},i}) = (v_{e_1,f} - v_{e_{k_i},f}) + (v_{e_2,f} - v_{e_1,f}) + \cdots + (v_{e_{k_i},f} - v_{e_{k_i-1},f}). \]
This sum telescopes to the zero vector as desired.


From this, we see that all independent cycle edge states are in $\ker(\Delta_-)$; indeed, for an independent cycle edge state $\alpha$, by the above we have
\[ \Delta_- \alpha = I^t I \alpha = 0, \]
so that $\alpha \in \ker(\Delta_-)$ as desired.

Lastly, recall that $\dim \ker(\Delta_-) = b_1$. Each independent cycle edge state is linearly independent, and moreover, we have precisely $b_1$ of these states, so that the kernel of $\Delta_-$ is equal to the span of independent cycle edge states. Thus by Theorem \ref{Theo: steady states in bijection with kernel}, we see that a state is steady if and only if it is a linear combination of independent cycle edge states as desired.
\end{proof}

\section{The graph Dirac operator and the Dirac equation}

\subsection{Matrix representations of graph Dirac operators}
We now proceed to consider the graph Dirac operators. These operators can be viewed as formal square roots of the graph Laplacians. Note that both the even and odd Laplacian are real symmetric matrices, hence are diagonalizable.

Let $Q_+ D_+ Q_+^{-1}$ and $Q_-, D_-, Q_-^{-1}$ be the diagonalizations of the even and odd Laplacians respectively. Let $\sqrt{D_{\pm}}$ be the matrix given by taking the square root of each diagonal entry. The graph Dirac operators are defined as follows.

\begin{definition}
The \textit{Even Dirac operator} is defined as
$\slashed{D}_+ = Q_+ \sqrt{D_+} Q_+^{-1}.$
\end{definition}

\begin{definition}
The \textit{Odd Dirac operator} is defined as
$\slashed{D}_- = Q_- \sqrt{D_-} Q_-^{-1}.$
\end{definition}

To see that the above operators serve as formal square roots of the even and odd Laplacians, note that we have
\[ (\slashed{D}_{\pm})^2 = Q_{\pm} \sqrt{D_{\pm}} Q_{\pm}^{-1} Q_{\pm} \sqrt{D_{\pm}} Q_{\pm}^{-1} = \Delta_{\pm}, \]
as desired.

\begin{definition}\label{def: Incidence_Dirac}
The \textit{Incidence Dirac operator} is defined as
$\slashed{D}_I = \begin{pmatrix} 0 & I \\ I^t & 0 \end{pmatrix}$.
\end{definition}

The incidence Dirac operator also serves as a formal square root of the Laplace operators, as we have

\[ (\slashed{D}_I)^2 = \begin{pmatrix} \Delta_+ & 0 \\ 0 & \Delta_- \end{pmatrix} .\]

In this section, we study the linear algebraic properties of the graph Dirac operators. In particular, we prove $\ker \slashed{D}_{\pm} = \ker \Delta_{\pm}$ and $\ker \slashed{D}_I = \ker \slashed{D}_+ \bigoplus \slashed{D}_-$. Additionally, we prove a relationship between the nonzero eigenvalues of $\slashed{D}_I$ and the nonzero eigenvalues of $\Delta_{\pm}.$ 


\begin{proposition}\label{Prop: Kernels_of_I_and_I^t}
$\ker{I^t} = \ker{\Delta^+}$ and $\ker{I} = \ker{\Delta^-}$.
\end{proposition}
\begin{proof}
 We prove the first equality; the second follows analogously. We first show $\ker I^t \subseteq \ker \Delta_+$. Indeed, suppose $v \in \ker I^t$. Then $\Delta_+ v = I (I^t v) = 0$, so $v \in \ker \Delta_+$. Conversely, suppose $v \in \ker \Delta_+$. Then $\Delta_+ v = I I^t v = 0$, so that $v^t (I I^t v) = (I^t v)^t (I^t v) = \| I^t v \| = 0$, so that $I^t v = 0$ as desired.
\end{proof}

\begin{proposition}\label{Prop:EvenKernel} $\slashed{D}_{+}$ and $\slashed{D}_{-}$ are non-negative definite matrices, $\ker{\slashed{D}_{+}}=\ker{\Delta_{+}}$, and $\ker{\slashed{D}_{-}}=\ker{\Delta_{-}}$.
\end{proposition}
\begin{proof}
We first claim that $\slashed{D}_+$ is non-negative definite. Let $QDQ^{-1}$ be the diagonalization of $\Delta_+$. Since $\Delta_+$ is non-negative definite, all entries in the matrix $D$ will be non-negative, so that all entries in the diagonal matrix $\sqrt{D}$ of $\slashed{D}_+ = Q \sqrt{D} Q^{-1}$ will be non-negative. Thus, the eigenvalues of $\slashed{D}_+$ are non-negative, so $\slashed{D}_+$ is non-negative definite as desired.

We next claim that $\ker{\slashed{D}_+} = \ker{\Delta_+}$. To this end, we show that $\ker \slashed{D}_+ \subseteq \ker \Delta_+$, and $\dim \ker \slashed{D}_+ = \dim \ker \Delta_+$.

To see the first of these claims, let $v \in \ker \slashed{D}_+$. By definition $\slashed{D}_+ v = 0$, so that $\Delta_+ v = \slashed{D}_+ (\slashed{D}_+ v) = 0$. Thus we have $v \in \ker \Delta_+$ as desired.

We now demonstrate the second claim. Let $\lambda_1, \lambda_2, \dots, \lambda_n$ be the eigenvalues of $\Delta_+$. By definition, we see the eigenvalues of $\slashed{D}_+$ are $\sqrt{\lambda_1}, \sqrt{\lambda_2}, \dots, \sqrt{\lambda_n}$.

Recall that $\dim \ker \Delta_+$ is precisely the geometric multiplicity of the eigenvalue $0$. From the above, we see that the algebraic multiplicities of $0$ for $\Delta_+$ and $\slashed{D}_+$ are equal.  Each of these matrices are diagonalizable, hence the geometric multiplicities of $0$ for $\Delta_+$ and $\slashed{D}_+$ are equal, so that $\dim \ker \Delta_+ = \dim \ker \slashed{D}_+$ as desired.

Similarly, $\slashed{D}_{-}$ is a non-negative definite matrix and $\ker{\slashed{D}_{-}}=\ker{\Delta_{-}}$.

\end{proof}

\begin{remark}
We know that $\ker \Delta_+ \cong \mathbb{C}^{b_0}$ and $\ker \Delta_- \cong \mathbb{C}^{b_1}$, so by the above we have $\ker \slashed{D}_+ \cong \mathbb{C}^{b_0}$ and $\ker \slashed{D}_- \cong \mathbb{C}^{b_1}$.
\end{remark}

\begin{proposition}
$\ker{\slashed{D}_I} = \ker{\Delta^+} \bigoplus \ker{\Delta^-}$.
\end{proposition}
\begin{proof}
We first prove that $\ker \slashed{D}_I = \ker I^t \bigoplus \ker I$. Let 
$v = 
\begin{bmatrix}
    x\\
    y
\end{bmatrix}$
where $x \in \mathbb{C}^{|V|}$ and $y \in \mathbb{C}^{|E|}$. By definition, we see 
\[ v \in \ker \slashed{D}_I \iff  \begin{bmatrix}
    0 & I\\
    I^t & 0
    \end{bmatrix}
    \begin{bmatrix}
    x\\
    y
    \end{bmatrix} = 0 \iff \begin{bmatrix} Iy \\ I^t x\end{bmatrix} = \begin{bmatrix} 0 \\ 0 \end{bmatrix}.\]
Thus we see $v = 
\begin{bmatrix}
    x\\
    y
\end{bmatrix} \in \ker \slashed{D}_I \iff x \in \ker I^t \text{ and } y \in \ker I$ so that by proposition $\ref{Prop: Kernels_of_I_and_I^t}$ we have
\[ \ker{\slashed{D}_I} = \ker{I^t} \bigoplus \ker{I}= \ker{\Delta^+} \bigoplus \ker{\Delta^-}. \]

\end{proof}

\begin{proposition}\label{prop: Dirac_Orientation}
The eigenvalues of $\slashed{D}_I$ are independent of orientation.
\end{proposition}

\begin{proof}
Let $\Gamma$ and $\Gamma'$ be two orientations of the same graph. Let $v$ be a vertex state on this graph, and let $e = \begin{pmatrix} e_1 \\ e_2 \\ \vdots \\ e_n \end{pmatrix}$ be an edge state on this graph. Without loss of generality, assume $\Gamma$ and $\Gamma'$ differ on edges $e_1, \cdots , e_k$. Suppose $\slashed{D}_I (\Gamma) \begin{pmatrix} v \\ e \end{pmatrix} = \lambda \begin{pmatrix} v \\ e \end{pmatrix}$, so that $\lambda$ is an eigenvalue of $\Gamma$. Then $I(\Gamma) e = \lambda v$ and $I^t (\Gamma) v = \lambda e$. By definition of the incidence matrix, for each $i \leq k$, we have col$_i (I(\Gamma')) = - \text{col}_i (I(\Gamma))$, so that $I(\Gamma') \begin{pmatrix} -e_1 \\ \vdots \\ -e_k \\ e_{k+1} \\ \vdots \\ e_n \end{pmatrix} = I(\Gamma) e = \lambda v$. \\


Because $I^t(\Gamma) v = \lambda e$, we know from the relationship between their columns that $I^t(\Gamma^{'}) v = \lambda \begin{pmatrix} -e_1 \\ \vdots \\ -e_k \\ e_{k+1} \\ \vdots \\ e_n \end{pmatrix}$. Thus, $\slashed{D}_I (\Gamma') \begin{pmatrix} v \\-e_1 \\ \vdots \\ -e_k \\ e_{k+1} \\ \vdots \\ e_n \end{pmatrix} = \lambda \begin{pmatrix} v \\ -e_1 \\ \vdots \\ -e_k \\ e_{k+1} \\ \vdots \\ e_n \end{pmatrix}$, so $\lambda$ is an eigenvalue for $\slashed{D}_I(\Gamma')$.
\end{proof}

\begin{corollary}
 The non-zero eigenvalues of $\slashed{D}_I$ are $ -\sqrt{\lambda_1}, \sqrt{\lambda_1}, -\sqrt{\lambda_2}, \sqrt{\lambda_2}, \cdots, -\sqrt{\lambda_n}, \sqrt{\lambda_n} $,  where $\lambda_1, \cdots, \lambda_n$ are the non-zero eigenvalues of $\Delta_\pm$.
\end{corollary}

\begin{proof}
Let $\lambda_1, \dots, \lambda_n$ be the nonzero eigenvalues of $\Delta_{\pm}$. Because $\slashed{D}_I^2 = \begin{bmatrix} \Delta_+ & O \\ O & \Delta_- \end{bmatrix}$, we have $P_{\slashed{D}_I^2} (t) = P_{\Delta_+}(t) \cdot P_{\Delta_-} (t)$, where $P_A(t)$ denotes the characteristic polynomial of a matrix $A$.

Because $\slashed{D}_I^2 = \begin{bmatrix} \Delta_+ & O \\ O & \Delta_- \end{bmatrix}$, we have $P_{\slashed{D}_I^2} (t) = P_{\Delta_+}(t) \cdot P_{\Delta_-} (t)$, where $P_A(t)$ denotes the characteristic polynomial of a matrix $A$. We know that $\Delta_+$ and $\Delta_-$ have the same eigenvalues with the same multiplicity, so the non-zero eigenvalues of $\slashed{D}_I^2$ are $\lambda_1, \lambda_1, \lambda_2, \lambda_2, \cdots, \lambda_n, \lambda_n$. Then the non-zero eigenvalues of $\slashed{D}_I$ are $\pm \sqrt{\lambda_1}, \pm \sqrt{\lambda_1}, \cdots , \pm \sqrt{\lambda_n}, \pm \sqrt{\lambda_n}$. \\
If $a$ is an eigenvalue of $\slashed{D}_I$, then $-a$ is an eigenvalue of $-\slashed{D}_I$. Because this is equivalent to the Incidence Dirac for the same graph with the opposite orientations, Proposition \ref{prop: Dirac_Orientation} tells us that $-a$ must also be an eigenvalue for $\slashed{D}_I$. This means that for each pair of $\pm \sqrt{\lambda_i}$ eigenvalues, if one is positive the other must be negative, and vice-versa. 
\end{proof}

\subsection{Visualization}

We implemented a program in Python which computes solutions to the different versions of the Dirac equation over time, given a random graph and initial random quantum state (vertex state for even Dirac operator, edge state for odd Dirac operator, vertex and edge state for incidence Dirac operator).

We plotted average position and average angle of these solutions over time, where average position is given as the average of real components and the average of imaginary components, and average angle is given as the angle formed by the coordinates of the average position in the complex plane.

Figure \ref{fig:2} is a typical plot of average position over time for solutions to the Dirac equation for the even, odd, and incidence Dirac operators, for a random graph with 10 vertices.

\begin{figure}[h]	
	\centering
	\begin{subfigure}{3.1in}
		\centering
		\includegraphics[width=3.2in]{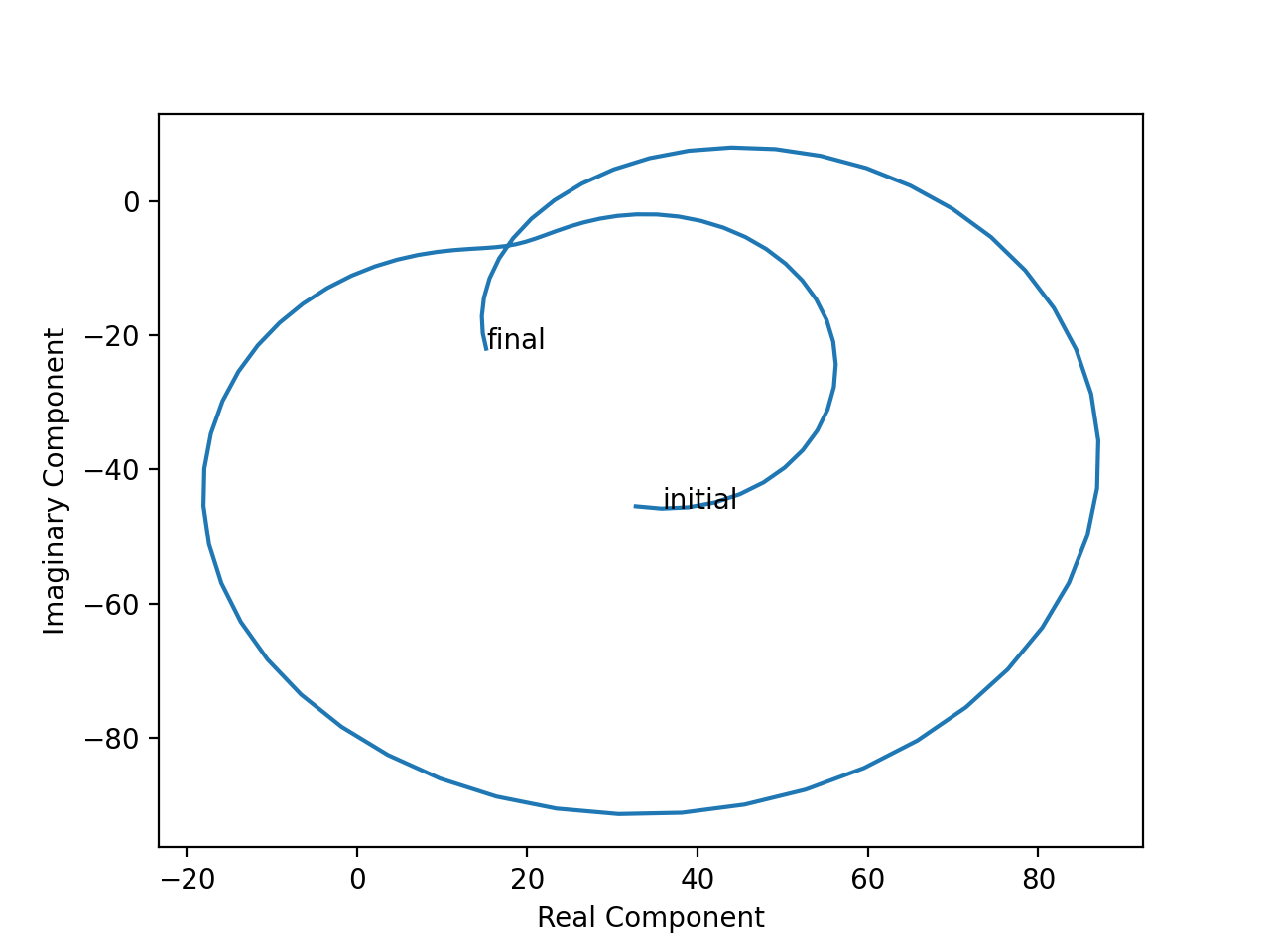}
		\caption{Even Dirac operator.}\label{fig:2a}		
	\end{subfigure}
	\begin{subfigure}{3.1in}
		\centering
		\includegraphics[width=3.2in]{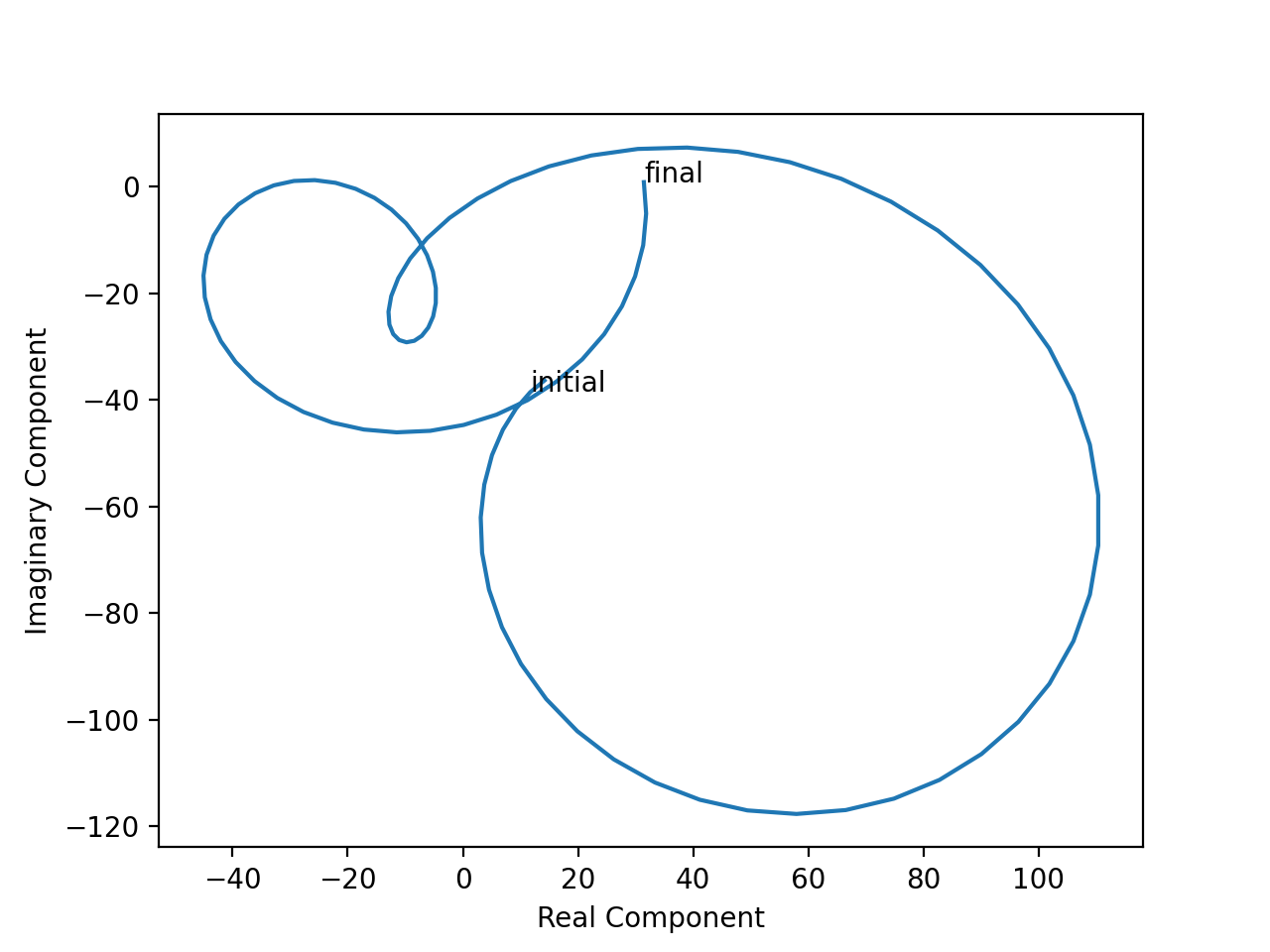}
		\caption{Odd Dirac operator.}\label{fig:2b}
	\end{subfigure}
	
	\begin{subfigure}{3.1in}
		\centering
		\includegraphics[width=3.2in]{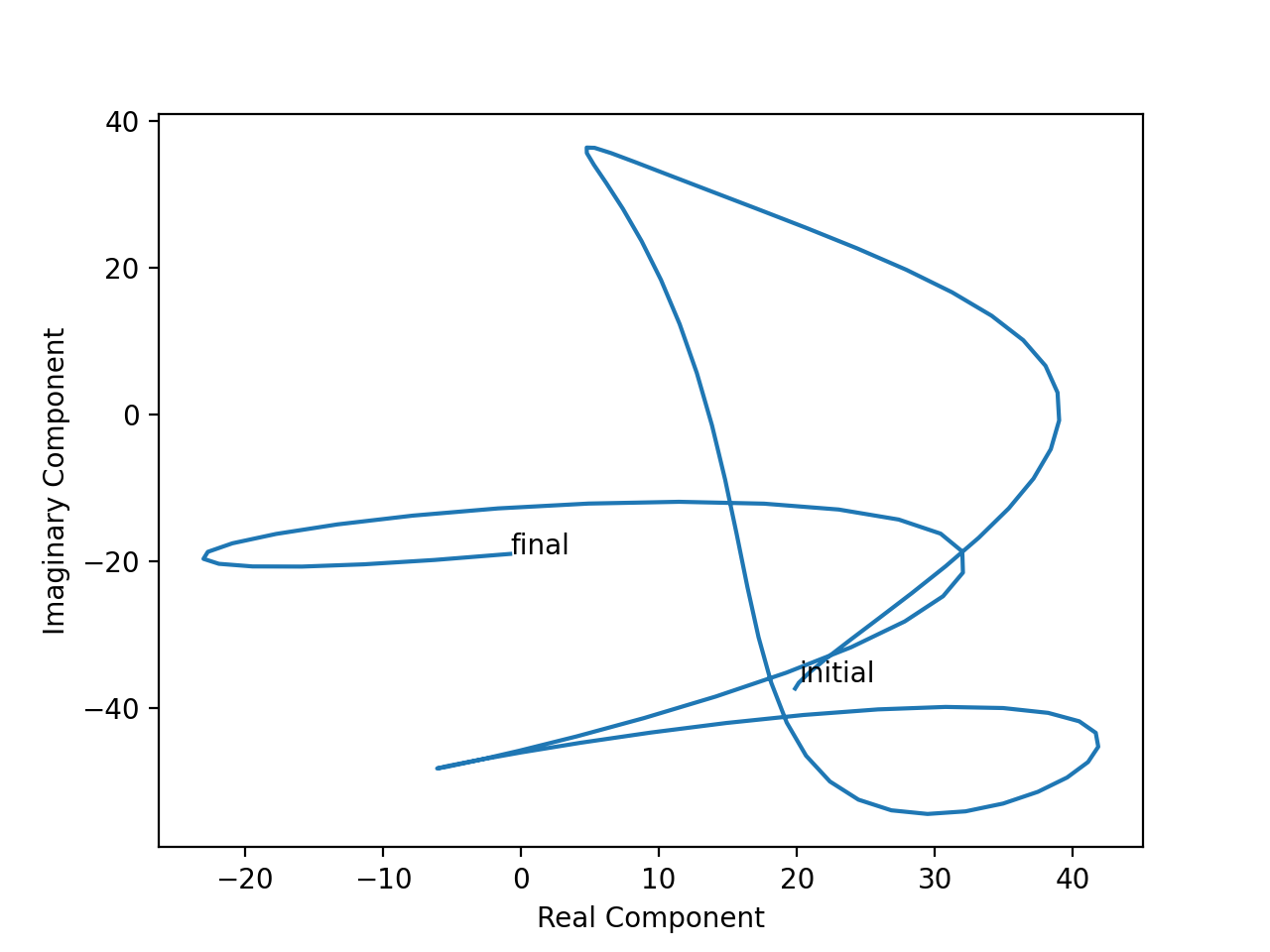}
		\caption{Incidence Dirac operator.}\label{fig:2c}
	\end{subfigure}
	\caption{Plots of average position over time for a random quantum state.}\label{fig:2}
\end{figure}

The geometric features of these plots can yield information to the behavior of the quantum system over time. As a brief example, consider a quantum state which is constant on each connected component of a graph. By corollary \ref{Even Schrodinger steady states}, we know that such a quantum state is steady over time. Figure \ref{fig:3} plots the average angle and position over time of a quantum state in the kernel of a graph. The nonconstant behavior of the graph is due to rounding error in Python.

\begin{figure}[h]
	\centering
	\begin{subfigure}{2.9in}
		\centering
		\includegraphics[width=2.8in]{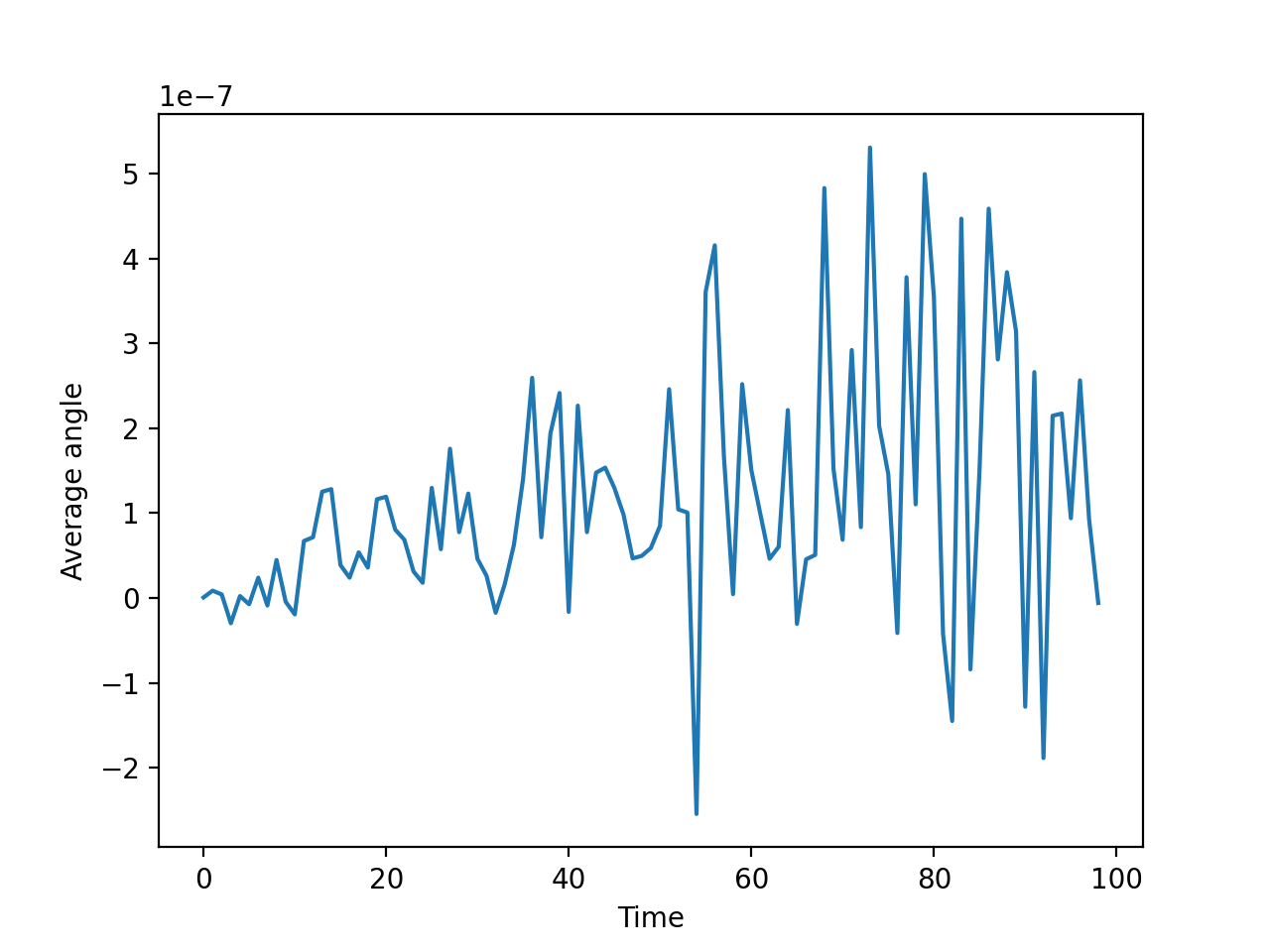}
		\caption{average angle over time.}\label{fig:3a}		
	\end{subfigure}
	\begin{subfigure}{2.9in}
		\centering
		\includegraphics[width=2.8in]{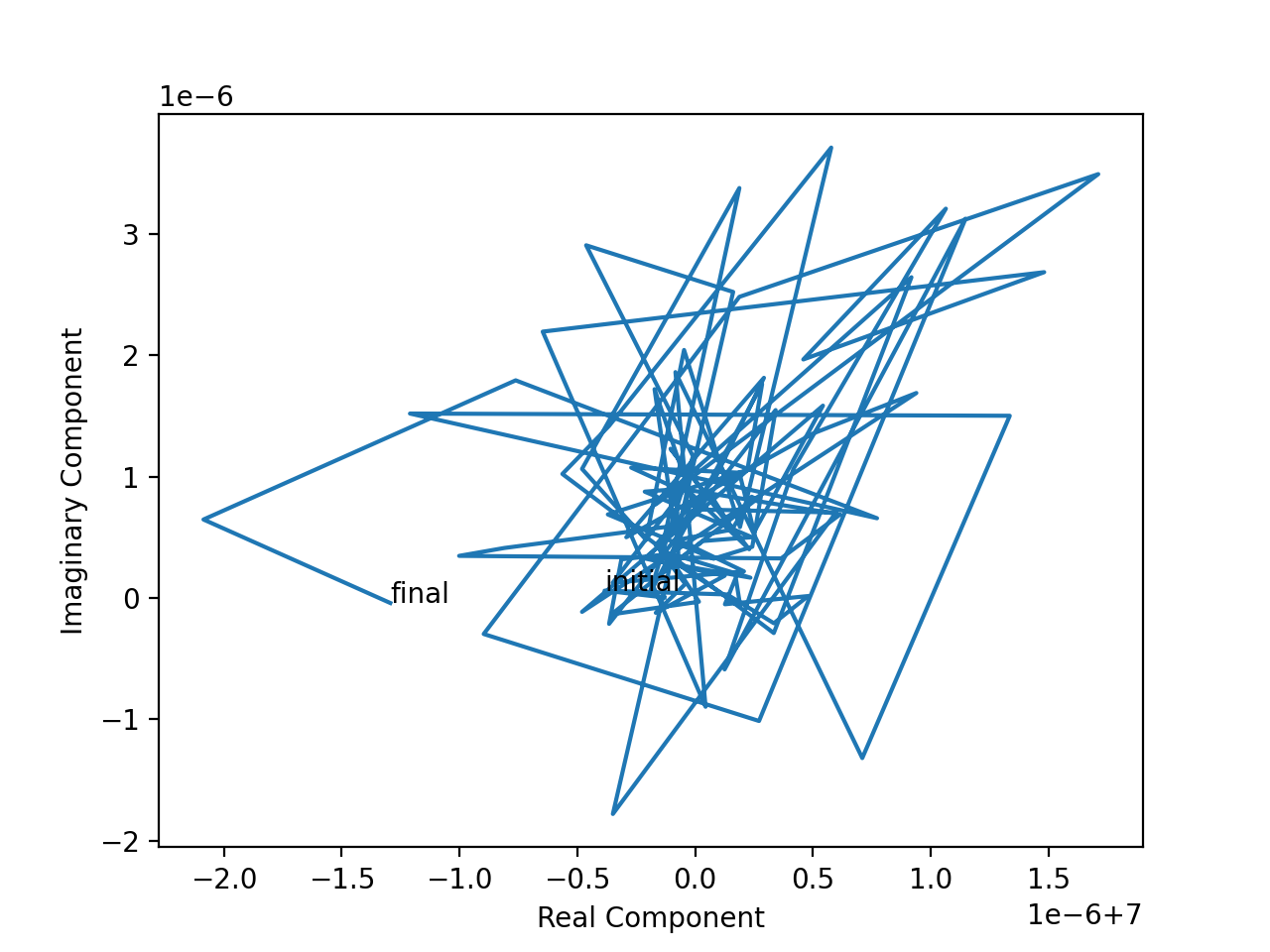}
		\caption{average position in complex plane.}\label{fig:3b}
	\end{subfigure}
	\caption{Plots of a vertex state in the kernel of the even Dirac operator.}\label{fig:3}
\end{figure}

\newpage
\subsection{Quadratic forms}

We can encode information on the eigenvalues of the graph operators via their quadratic forms, as their signatures provide invariants of the graph, which determine the behavior of the solutions of the graph Schr\"odinger and Dirac equations. \\
\begin{definition}
For a symmetric matrix $M$, its associated quadratic form is defined as $q_M(v)=v^t M v$
\end{definition}

Using this definition we find graph theoretical expressions for the quadratic forms of our graph operators. For instance, it is well known \cite{nica_SpectralGraph} that the quadratic form of $\Delta_+$ is 

\begin{equation}
q_{\Delta_+} (v) = \sum_{e_{i,j} \in E(\Gamma)} (v_j - v_i)^2.
\end{equation}

We find similar expressions for the quadratic forms of other operators. Where the even Laplacian sums a function on the vertices over the edges, the odd Laplacian sums a similar function on the edges over the vertices.

\begin{proposition}
The quadratic form of the odd Laplacian is
\[
q_{\Delta_-}(e) = \sum_{v_i \in V(\Gamma)}\big( \sum_{e_{i,j} \in E(\Gamma)} k_{i,j} \cdot e_{i,j} \big)^2,
\]

where $k_{i,j} = \begin{cases} -1 & e_{i,j} \text{ leaves } v_i \\ 1 & e_{i,j} \text{ enters } v_i  \end{cases}.$
\end{proposition}

    \begin{proof}
    From the definition of the odd Laplacian, we have that
    \[
    \Delta_-(i,j) = \begin{cases} 2 & i = j \\
    0 & i \neq j, e_i \text{ not incident to } e_j \\
    1 & i \neq j, e_i \text{ incident to } e_j, \text{ both start or end at } v_s \\
    -1 & i \neq j, e_i \text{ incident to } e_j, \text{ one starts and the other ends at } v_s. \\
    \end{cases}
    \]
    For each $v_s \in V(\Gamma)$, there are $ \frac{D(v_s)(D(v_s)-1)} 2$ terms of the form $(-1)^k 2 e_{i,s} \cdot e_{j,s}$, where $k = 0$ when $e_i$ and $e_j$ both start or end at $v_s$ and $k= 1$ when one starts and the other ends at $v_s$. The diagonal adds the term $2 e_{i,j}^2$ to the quadratic form. Each edge is incident to exactly two vertices, thus we assign a squared term to each of these vertices. Thus for each $v_s$, we have \[\sum_{e_{i,s} \in E(\Gamma)} \left (e_{i,s}^2 + \sum_{e_{j,s}\neq e_{i,s} \in E(\Gamma)} (-1)^k 2 e_{i,s} \cdot e_{j,s}\right )\] in this quadratic form. Because each $e_{i,s}$ which appears in the second sum will appear in our first,  this sum is equal to \[\left (\sum_{e_{i,s} \in E(\Gamma)} k_{i,s} e_{i,s} \right )^2,\] where $k=-1$ if an edge starts at $v_s$ and $k = 1$ if it ends at $v_s$. \\
    Adding these terms for each vertex finishes the proof.
    \end{proof}

\begin{proposition}
The quadratic form of the incidence Dirac operator $\slashed{D}_I$  is 
\[
q_{\slashed{D}_I} \begin{pmatrix} v \\ e \end{pmatrix} = 2 \sum_{e_{i,j} \in E(\Gamma)} e_{i,j} \cdot (v_j - v_i).
\]
\end{proposition}

    \begin{proof}
    Because $(v^t \ e^t) \slashed{D}_I \begin{pmatrix} v \\ e\end{pmatrix} = v^t I e + e^t I^t v$, and these $1 \times 1$ matrices are transposes of one another, they are equal, giving $(v^t \ e^t) \slashed{D}_I \begin{pmatrix} v \\ e\end{pmatrix} = 2 e^t I^t v$. Each row of $I^t$ represents some $e_{i,j}$, and returns $v_j - v_i$ (where $e_{i,j}$ goes from $v_j$ to $v_i$), so when we multiply this out, we get $\sum_{e_{i,j} \in E(\Gamma)} e_{i,j} \cdot (v_j - v_i)$. 
    \end{proof}

Because the roots of the quadratic form of a non-negative definite matrix are exactly the vectors in the kernel of the matrix, the set of roots of $q_{\Delta_\pm}$ (denoted by root$(q_{\Delta_\pm})$) provides important graph-theoretical information. Specifically, $q_{\Delta_+} (v)= 0$ if and only if $v$ is constant on connected components, and $q_{\Delta_-}(e) = 0$ if and only if $e$ represents a cycle on the graph. Interestingly, there is not a similar equality for $\slashed{D}_I$, as stated in the following proposition. \\

\begin{proposition}
\[
\text{root }(q_{\slashed{D}_I}) \supseteq \ker(\Delta_+) \bigoplus \mathbb{C}^{|E|} \cup \mathbb{C}^{|V|} \bigoplus \ker(\Delta_-).
\]
\end{proposition}

\begin{proof}
If $v \in \ker(\Delta_+)$, then we know $v$ is constant on connected components, so $v_j - v_i = 0$ for each $e_{i,j} \in E(\Gamma)$. Then for any $e$, we will have that $q_{\slashed{D}_I} (v \bigoplus e) = 0$. 
\\
For $e \in \ker(\Delta_-)$, any non-zero terms of $e$ must be a part of a cycle, so by definition of $q_{\slashed{D}_I}$, each $v_i$ will be added and subtracted exactly once on the cycle (if an edge goes in the opposite direction to an incident edge, the edge term must have the opposite sign to be in the kernel). From the properties of $\ker(\Delta_-)$, the terms by which these vertex terms are multiplied will be equal, so $v \bigoplus e$ is a root of our quadratic for any $v$.

\end{proof}

Because the eigenvalues of the Incidence Dirac operator come in pairs of positive and negative values, it is not non-negative definite, so the root of the quadratic form doesn't match with the kernel of the matrix as it did for the Laplace operators. For this reason, it is not clear what the entire root of $q_{\slashed{D}_I}$ would look like for most graphs. \\
An example of a graph where we can see the entire root of $q_{\slashed{D}_I}$ is $P_2$, for which we have 
\[
\slashed{D}_I = \begin{pmatrix} 0 & 0 & -1 \\ 0 & 0 & 1 \\ -1 & 1 & 0 \end{pmatrix}.
\]
Here, if $v \bigoplus e \in {\rm root}(q_{\slashed{D}_I})$, we must have that either $e = 0$ or $v_1 = v_2$, so in this case $\text{root }(q_{\slashed{D}_I}) =  \ker(\Delta_+) \bigoplus \mathbb{C}^{|E|} \cup \mathbb{C}^{|V|} \bigoplus \ker(\Delta_-) $
\\
By contrast, cycle graphs will have many roots which are not a part of this union. For example, the vector $\begin{pmatrix} 1 \\2 \\2\\2\\1\\2 \end{pmatrix}$ is a root for $q_{\slashed{D}_I}$ for the oriented graph in Figure \ref{fig:cycle}.

\begin{figure}[h]
    \centering
    \includegraphics[scale=0.4]{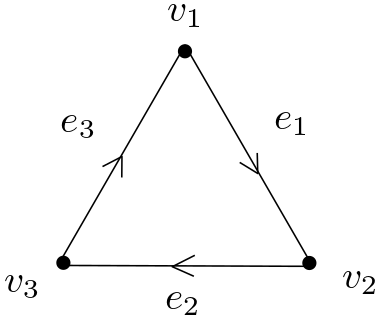}
    \caption{Oriented cycle graph $C_3$.}\label{fig:cycle}
\end{figure}

\subsection{Powers of $\slashed{D}_I$} 
An interesting property of the Laplace operators is the fact that the entries of the powers of these matrices encode combinatorial information about superwalks on the graph, as explained by \cite{Yu17}. \\
A \textit{vertex superwalk} of length 1 starting at $v_i$ can 
\begin{enumerate}
    \item go through an incident edge to one of it's neighbors, in which case sgn$(\gamma) = -1$, or
    \item go through an incident edge and return to $v_i$, in which case sgn$(\gamma) = 1$.
\end{enumerate}

Similarly, an \textit{edge superwalk} of length 1 starting at $e_i$ can
\begin{enumerate}
    \item go through an incident vertex $v$ to an adjacent edge $e_j$, in which case sgn$(\gamma) = 1$ if $e_j$ and $e_i$ both go into or both leave $v$ and sgn$(\gamma)=-1$ if one of $e_j$ and $e_i$ enters $v$ and the other exits it, or
    \item go to an incident vertex and back to $e_i$ in which case sgn$(\gamma) = 1$.
\end{enumerate}

Similarly, we can define a new type of walk on the graph which is encoded by the Incidence Dirac operator.
\begin{definition}
A \textit{vertex-edge walk} of length 1 moves from a vertex to an incident edge, or an edge to an incident vertex. We can define the sign of the walk as follows: If the edge involved in the walk originates from the vertex involved then  sgn$(\gamma) = -1$ and if the edge enters the vertex involved then sgn$(\gamma)=1$. Note for this definition we do not care whether we're starting at a vertex or an edge. For a graphical way of thinking about the signs of these walks see figure \ref{fig:sgn}. \\
A vertex-edge walk of length $k$ is a combination of $k$ walks of length 1, and has a sign equal to the product of the signs of each step.
\end{definition}

\begin{figure}[h]
    \centering
    \includegraphics[scale = 0.4]{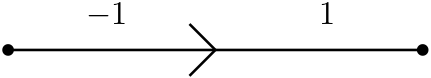}
    \caption{Signs of vertex-edge walks on a $P_2$ graph.}
    \label{fig:sgn}
\end{figure}
By this definition, we note that vertex superwalks are vertex-edge walks with an even length, while edge superwalks are vertex-edge walks of an odd length.
\begin{theorem} \label{thm: Incidence_Dirac}
\[\slashed{D}_I^k (i,j) = \sum_{\gamma, i \to j, k} \text{sgn}(\gamma).\]
\end{theorem}

\begin{proof}
We will prove this theorem by induction on $k$. \\
Letting $k=1$, we know $\slashed{D}_I = \begin{pmatrix} O & I \\ I^t & O \end{pmatrix}$, and checking these entries we see that the theorem is true. \\
Let it be true for some arbitrary $k-1$. There are two cases. \\ Case 1: Let $k-1$ be odd. Then $\slashed{D}_I^k = \begin{pmatrix} \Delta_+^{k/2} & O \\ O & \Delta_-^{k/2} \end{pmatrix}$, so we know that the entries count vertex superwalks and edge superwalks of length $k/2$, which means that they count vertex-edge walks of length $k$. \\
Case 2: Let $k-1$ be even. Then $\slashed{D}_I^k = \begin{pmatrix} O & \Delta_+^{(k-1)/2} I \\ \Delta_-^{(k-1)/2} I^t & O \end{pmatrix}$. We know there cannot be vertex or edge superwalks, so the $0$ block matrices are expected. Because the other block matrices are transposes of one another, we will analyze $A = \Delta_+^{(k-1)/2} I$. Looking at each entry, we have \[a_{i,j} = row_i (\Delta_+^{(k-1)/2}) \cdot col_j (I) = 1 \cdot \sum_{\gamma, i \to n, k-1} \text{sgn}(\gamma) -1 \cdot \sum_{\gamma, i \to m, k-1} \text{sgn}(\gamma),\] where $e_j$ goes from $v_m$ to $v_n$. Then because the last step of a vertex-edge walk starting at $v_i$ and ending at $e_j$ must include one of the vertices adjacent to $e_j$, our sum must include a term for each such walk, and we have that $a_{i,j} = \sum_{\gamma, i \to j, k} \text{sgn}(\gamma)$. 
\end{proof}

As an example of this result we look to the third power of the incidence matrix for $K_3$ with clockwise orientation. 
$$\slashed{D}_I^3 = \begin{pmatrix} 0 & 0 & 0 & 0 & 3 & -3 \\ 0 & 0 & 0 & -3 & 0 & 3 \\ 0 & 0 & 0 & 3 & -3 & 0 \\ 0 & -3 & 3 & 0 & 0 & 0 \\ 3 & 0 & -3 & 0 & 0 & 0 \\ -3 & 3 & 0 & 0 & 0 & 0   \end{pmatrix}.$$ \\

Here we see that there are 3 vertex-edge walks between a vertex and and an incident edge, and that the powers follow as expected. Looking at the entries between $v_1$ and $e_1$, we see a value of 0, because there are two such walks of length 3 connecting these elements, but these walks have opposite signs, cancelling one another out as seen in Figure \ref{fig:F8}. 
\begin{figure}[h]
    \centering
    \includegraphics[scale = 0.3]{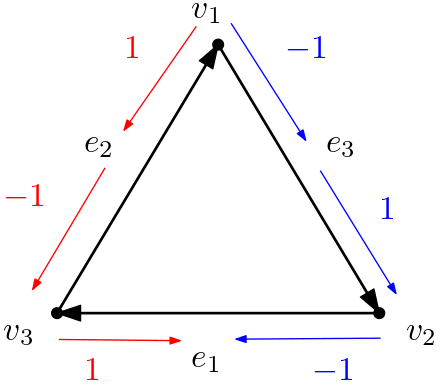}
    \caption{Length 3 walks from $v_1$ to $e_1$.}
    \label{fig:F8}
\end{figure}

The full list of vertex edge walks and their signs for $K_3$ originating at $v_1$, each of which is equivalent to performing the same moves at one of the other vertices.
\begin{align*}
    \gamma_1 &= v_1 \to e_2 \to v_3 \to e_1 \ \ \ \ \ \ & sgn(\gamma_1)=-1 \\
    \gamma_2 &= v_1 \to e_3 \to v_2 \to e_1 \ \ \ \ \ \ & sgn(\gamma_2)=1 \\
    \gamma_3 &= v_1 \to e_2 \to v_3 \to e_2 \ \ \ \ \ & sgn(\gamma_3) = 1 \\
    \gamma_4 &= v_1 \to e_2 \to v_1 \to e_2 \ \ \ \ \ & sgn(\gamma_4)=1 \\
    \gamma_5 &= v_1 \to e_3 \to v_1 \to e_2 \ \ \ \ \ & sgn(\gamma_5)=1 \\
    \gamma_6 &= v_1 \to e_3 \to v_2 \to e_3 \ \ \ \ \ & sgn(\gamma_6) = -1 \\
    \gamma_7 &= v_1 \to e_3 \to v_1 \to e_3 \ \ \ \ \ & sgn(\gamma_7)=-1 \\
    \gamma_8 &= v_1 \to e_2 \to v_1 \to e_3 \ \ \ \ \ & sgn(\gamma_8)=-1. \\
\end{align*}


\section{Dimer models and a gluing formula for Dirac operators}\label{sec:Dimer}
Gluing formulae for discrete operators have been studied in mathematical physics \cite{Reshetikhin:2014jaa}, and in the case of the discrete Laplace operator, it has some interesting connections with graph theory \cite{Contreras20}. On the other hand, dimer models have been studied extensively in the context of statistical physics, dimer models, quantum mechanics and combinatorics \cite{Cimasoni07, KenyonNotes, Kenyon2002}.

In this section we describe a combinatorial interpretation of a gluing formula for the Kastelyn matrix, which can be regarded as a discrete Dirac operator for lattice graphs.

\subsection{Graph gluing}
First, we define two different types of gluing of graphs, which are relevant while discussing gluing identities for graph Dirac operators.
\begin{definition}\label{Def:interface_graph}
Let $\Gamma_1$ and $\Gamma_2$ be two graphs. Let $\Gamma^\partial_1$ and $\Gamma^\partial_2$ be isomorphic subgraphs of $\Gamma_1$ and $\Gamma_2$ respectively. Then $I=\Gamma^\partial_1=\Gamma^\partial_2$ is an \textit{interface} of $\Gamma_1$ and $\Gamma_2$.
\end{definition}

\begin{definition}\label{Def:interface_gluing}
Let $\Gamma_1$ and $\Gamma_2$ be two graphs and $I$ an interface. The \textit{interface gluing} of the two graphs  $\Gamma_1 \sqcup_I \Gamma_2$ is defined by: $V(\Gamma_1 \sqcup_I \Gamma_2)= (V(\Gamma_1) \setminus V(I)) \cup (V(\Gamma_2) \setminus V(I)) \cup V(I)$ and $E(\Gamma_1 \sqcup_I \Gamma_2)= (E(\Gamma_1) \setminus V(I)) \cup (E(\Gamma_2) \setminus E(I)) \cup E(I)$.
\end{definition}

\begin{definition}\label{Def:bridge_graph}
Let $\Gamma_1$ and $\Gamma_2$ be two graphs. Let $\{v_{1,1}, ..., v_{1,k}\} \subseteq \Gamma_1$ and $\{v_{2,1}, ..., v_{2,k}\} \subseteq \Gamma_2$. Let $\{e_1, ..., e_k\}$ be the set of edges that connect the pairs of vertices $(v_{1,1},v_{2,1}), ..., (v_{1,k},v_{2,k})$. Then a \textit{bridge graph} $B$ between $\Gamma_1$ and $\Gamma_2$ is the graph with $V(B)=\{v_{1,1}, ..., v_{1,k}, v_{2,1}, ..., v_{2,k}\}$ and $E(B)=\{e_1, ..., e_k\}$.
\end{definition}

\begin{definition}\label{Def:bridge_gluing}
Let $\Gamma_1$ and $\Gamma_2$ be two graphs and $B$ a bridge graph. The \textit{bridge gluing} of the two graphs  $\Gamma_1 \sqcup_B \Gamma_2$ is defined by: $V(\Gamma_1 \cup_B \Gamma_2)=V(\Gamma_1) \cup V(\Gamma_2)$ and $E(\Gamma_1 \sqcup_B \Gamma_2)=E(\Gamma_1) \cup E(B) \cup E(\Gamma_2)$.
\end{definition}
A \textit{perfect matching} of a graph is a set of edges such that each vertex is connected to exactly one edge. Kasteleyn's Theorem \cite{Kasteleyn} allows us to connect the number of perfect matchings to the Kasteleyn matrix, $K$, defined as the weighted adjacency matrix of a graph $G$, with horizontal edges weighted 1 and vertical
edges weighted $i =\sqrt{-1}$. We are interested in this matrix because of its connection to graph quantum mechanics: following \cite{KenyonNotes}, the matrix $K$ acts like a Dirac operator when restricted to sublattices of a lattice graph. That is, when a lattice graph's Kastelyn matrix is restricted to the four sublattices comprised of vertices spaced two apart in each cardinal direction, $K^*K=\Delta_+$.
\\\\
\begin{figure}[h]
    \centering
    \includegraphics[scale=.2]{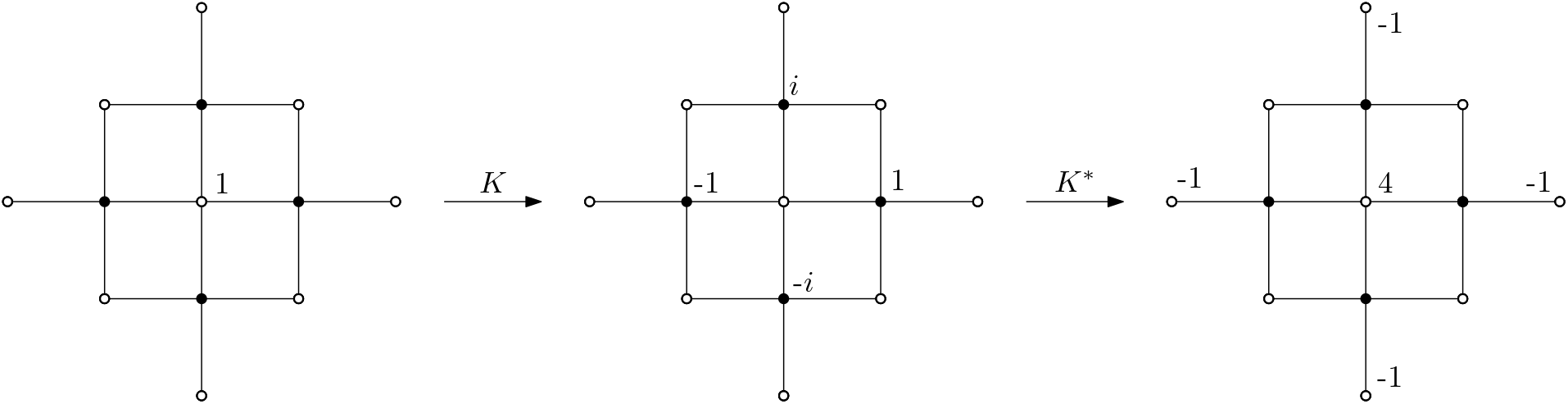}
    \caption{Visual representation of Kasteleyn operating on graph as a Laplacian.}
    \label{fig:kasteleyn}
\end{figure}


We reinterpret Kasteleyn's theorem by introducing explicit recurrence relations for perfect matchings of integer lattices. We prove gluing identities for the determinants of Kasteleyn matrices by introducing explicit formulae for perfect matchings of integer lattices when one of the sides is 2, 3, or 4. We denote an $k$ by $n$ lattice graph as $L_{k,n}$ and the number of ways to tile $L_{k,n}$ as $T_k(n)$.

\subsection{The $2\times n$ case}
\begin{theorem}\label{Theo:T9}
The number of unique ways that dominos can tile $L_{2,n}$ is given by the recursive formula $T_2(n) = T_2(n-1) + T_2(n-2)$, where $T_2(1) = 1$ and $T_2(2) = 2$.
\end{theorem}

	    \begin{proof}
    	For clarity of notation, note that $L_{2,n}$ has 2 rows of vertices and $n$ columns of vertices (so here a vertical domino fills one column of the graph). \\ 
    	We can see that when $n = 1$, $T_2(n) = 1$, as it can only be filled by placing the domino vertically, and when $n = 2$, $T_2(n) = 2$, as it can be filled by placing both tiles vertically or both tiles horizontally \\
    	We will now prove that for all $n > 2$, $T_2(n) = T_2(n-1) + T_2(n-2)$. \\
    	Let $n > 2$ be arbitrary. We can split the possible ways of filling our graph into 2 cases, depending on the orientation of the tile(s) in the first column \\ 
	    	Case 1: Let the first column be filled by a vertically placed domino. Then we have a $L_{2,n-1}$ graph left to fill, so from the definition of our function we have that there are $T_2(n-1)$ unique ways to tile the graph. \\
	    	Case 2: Let the first 2 columns be filled by 2 horizontally placed dominos. Then we have a $L_{2,n-2}$ graph left to fill, so there will be $T_2(n-2)$ unique ways to tile the graph. \\
		Now because each possible way to fill the graph is given by one of these two cases, $T_2(n)$ will be the sum of the number of ways to fill the graph in each of these cases, so we have that $T_2(n) = T_2(n-1) + T_2(n-2)$.
    	\end{proof}
	\begin{figure}[h]
    \centering
    \includegraphics[scale=0.15]{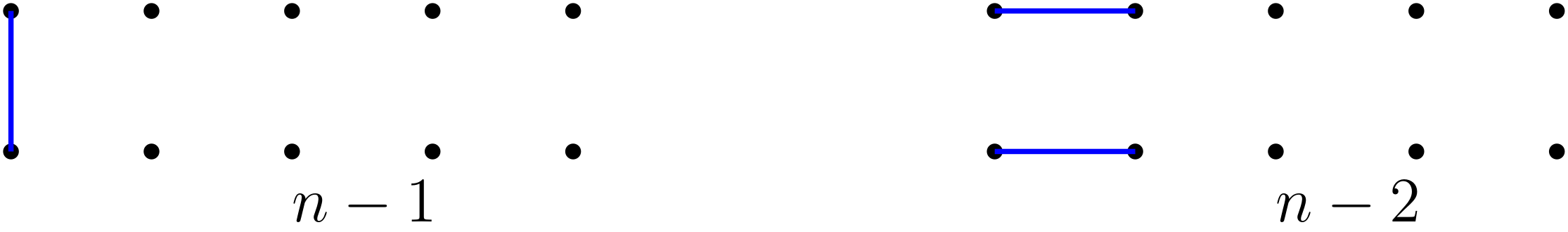} 
    \caption{Cases for end of a $2 \times n$ graph.}
    \end{figure}
    
    We can also find the formula for the number of perfect matchings for two such graphs glued together. In order to talk about this in more general terms, we will introduce notation for such a formula, given by $T_k(m \sqcup_{s,B} n)$, where we are gluing $L_{k,m}$ and $L_{k,n}$ graphs such that their sides of length $k$ are aligned, with $s$ indicating by how many vertices the right graph has been shifted down relative to the left, as seen in \ref{fig:G3,1}, and $B$ being the set of bridges from the gluing included in the perfect matching. $s$ and $B$ aren't extremely impactful in the $2 \times n$ case, but make a big difference with bigger graphs. 
    \begin{theorem}\label{Theo:G2,2}
The number of perfect matchings of $L_{2,m}$ and $L_{2,n}$ glued together with a shift is given by 
\[
T_{2}(m \sqcup_{s,B} n) = \begin{cases}
    T_2(m)T_2(n) & s=1 \text{ or } 2, B= \emptyset\\
    0 & s=0 \text{ or } 1, B= \{e_1\} \\
    T_2(m)T_2(n) + T_2(m-1)T_2(n-1) & s=0, B= \emptyset \text{ or } \{e_1, e_2\}. 
\end{cases}\\
\]
\end{theorem}
    \begin{proof}
    We begin by noting that in the case of a shift 1 or 2 gluing, it is only possible to tile the graph if we treat them separately, giving rise to the first equation. In general, for $T_k(m \sqcup_{s,B} n)$ where $k$ is even and $|B|$ is odd, we find that there are 0 ways to tile because that would leave an odd number of spaces to tile on either side of the gluing. \\
    For the third case, letting $m$ and $n$ be arbitrary, we want to find how many ways we can tile the $2 \times (m+n)$ graph formed by gluing  $2 \times m$ and $2 \times n$ graphs together. We can split this into two sub-cases: that where we include both bridges in the matching and that where neither are included. \\
    Case 1: Let there be no bridges between the graphs. Then for each possible tiling of the $2 \times m$ graph, there are $T_2(n)$ ways to finish the tiling. Thus, in this case there are $T_2(m)T_2(n)$ solutions. \\
    Case 2: Let there be bridges between the graphs. Then we have a $2 \times (m-1)$ and a $2 \times (n-1)$ graph left to fill, so there will be $T_2(m-1)T_2(n-1)$ solutions. This completes the proof.
    \end{proof}
    
\subsection{The $3\times n$ case}

\begin{theorem}\label{Theo:T9}
The number of unique ways that dominos can tile $L_{3,n}$ is given by the recursive formula $T_3(n) = 4T_3(n-2) - T_3(n-4)$, where $T_3(0) = 1$, $T_3(1) = 0$, $T_3(2) = 3$, and $T_3(3) = 0$. 
\end{theorem}
\begin{proof}
If n is odd, then the total number of spaces is odd, so it is impossible to cover the graph using domino tiles, each of which covers an even number of spaces. In order to have a nonzero number of ways to tile the graph, n must be even, so $n=2k$ for some $k\in \ZZ$. Looking at one end of the graph, and considering a vertical line that separates the first two columns from the rest, any tiling will either have no dominoes that cross the line or it will have some dominoes that cross the line. In the first case, there are $T_3(2)T_3(n-2)=3T_3(n-2)$ ways. In the other case, there are some domino tiles that cross the vertical line separating the second and third columns. Now consider a vertical line between the fourth and fifth columns. There are only two ways to tile $L_{3,4}$ with at least one horizontal domino on the second and third columns, adding $2T_3(n-4)$ ways to tile the whole graph. Similarly, pushing the vertical line further down the graph $2m$ spaces to account for all the possible tilings only adds $2T_3(n-2m)$ to the total, so the number of ways to tile $L_{3,n}$ is equal to $3T_3(n-2)+2(T_3(n-4)+T_3(n-6)+...+T_3(0))$. Since $2(T_3(n-4)+T_3(n-6)+...+T_3(0))=T_3(n-2)-T_3(n-4)$, we can simplify $T_3(n)$ to the recurrence relation $T_3(n)=4T_3(n-2)-T_3(n-4)$. \\
\end{proof}

\begin{figure}[h]
    \centering
    \includegraphics[scale = 0.4]{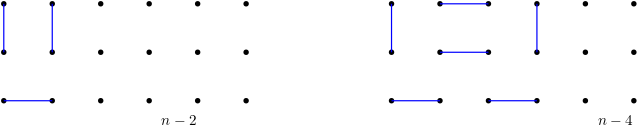}
    \caption{First 2 cases for tiling from the end of the graph.}
    \label{fig:3nrecursivecases}
\end{figure}
\begin{proposition}\label{Prop:T3sum} 
\[\sum_{i=0}^{\frac{n-2}{2}}T_3(2i)=\frac{T_3(n)-T_3(n-2)}{2}.\]
\end{proposition}
\begin{proof}
Recall from the proof of \ref{Theo:T9} that $T_3(n)=3T_3(n-2)+2\sum_{i=0}^{\frac{n-4}{2}}T_3(2i)$.  The result follows from subtracting $T_3(n-2)$ from both sides and dividing by 2.
\end{proof}
We can also find information about gluing these graphs together, first considering the case where the lattice graphs are shifted such that there is only one bridge between the two, such as in Figure \ref{fig:G3,1}. 
\begin{figure}[h]
    \centering
    \includegraphics[scale = 0.5]{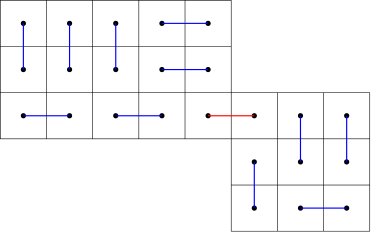}
    \caption{A perfect matching of $L_{3,5} \sqcup_{2,\{e_1\}} L_{3,3}$.}
    \label{fig:G3,1}
\end{figure}
\begin{theorem}\label{Theo:T10}
The number of perfect matchings of $L_{3,m}$ and $L_{3,n}$ glued together with a shift is given by 
\[
T_{3}(m \sqcup_{s,B} n) = \begin{cases}
    T_3(m)T_3(n) & B= \emptyset\\
    \frac{(T_3(m+1)-T_3(m-1))(T_3(n+1)-T_3(n-1))}{4} & s=2, B= \{e_1\}\\
    0 & s=1, B= \{e_1\} \text{ or } \{e_2\}\\
    \frac{(T_3(m) - T_3(m-2))(T_3(n)-T_3(n-2))}{4} & s=1, B= \{e_1,e_2\}\\
    \frac{(T_3(m+1)-T_3(m-1))(T_3(n+1)-T_3(n-1))}{4} & s=0, B= \{e_1\} \text{ or } \{e_3\}\\
    0 & s=0, B= \{e_2\} \text{ or } \{e_1, e_3\}\\
    \frac{(T_3(m) - T_3(m-2))(T_3(n)-T_3(n-2))}{4} & s=0, B= \{e_1,e_2\} \text{ or } \{e_2, e_3\}\\
    T_3(m-1)T_3(n-1) & s=0, B= \{e_1,e_2,e_3\}.\\
\end{cases}\\
\]
\end{theorem}

\begin{proof}
When $B=\emptyset$, there are no bridges between the two lattice graphs, so the number of ways to tile them is the number of ways to tile one of them multiplied by the number of ways to tile the other. That is, $T_3(m)T_3(n)$.\\\\

When $s=2$ and $B=\{e_1\}$ or when $s=0$ and $B=\{e_1\}$ or $\{e_3\}$, then the graphs are positioned like in Figure \ref{fig:G3,1}. If $m$ or $n$ is even, then there will be an odd number of vertices left on either the $3 \times m$ or $3 \times n$ graph which is impossible to tile, so suppose both $m$ and $n$ are odd. The method of solving will be the same for each graph in the gluing, so we will look at the $3 \times n$ graph. On the first row, there are 2 available vertices adjacent to one another, so we can split this into two cases, as seen in Figure \ref{fig:intermediate_tiling_configurations}: either include the vertical edge between the two, in which case there are $T_3(n-1)$ ways to finish the perfect matching, or we include the horizontal edges connecting them to the adjacent row. In this case we're forced to include the edge next to the bridge going into the top or bottom of the third edge. This means we have a $3 \times (n-2)$ lattice graph missing one vertex, so we will finish tiling it in the same way as we tiled the $3 \times n$ graph missing a vertex. Then we will add a $T_3 (n-3)$ term and continue until we have a $2 \times 1$ graph left, adding $1 = T_3(0)$ to the sum. Thus we get that the number of ways to tile $L_{3,n}$ missing one corner vertex is equal to $T_3(n-1)+T_3(n-3)+...+T_3(0)$. By proposition \ref{Prop:T3sum}, this sum equals $\frac{T_3(n+1)-T_3(n-1)}{2}$. Similarly, the number of ways to tile the left graph is $\frac{T_3(m+1)-T_3(m-1)}{2}$. Multiplying these two values together gives us the total number of tilings for these shift and bridge gluing combinations: $\frac{(T_3(m+1)-T_3(m-1))(T_3(n+1)-T_3(n-1))}{4}$.

\begin{figure}[h]
    \centering
    \includegraphics[scale = 0.6]{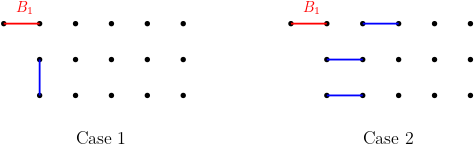}
    \caption{Possible ways to finish a perfect matching including $B_1$.}
    \label{fig:intermediate_tiling_configurations}
\end{figure}

When $s=1$ and $B=\{e_1\}$ or $\{e_2\}$ or when $s=0$ and $B=\{e_2\}$ or $\{e_1, e_3\}$, then at least one of the $3 \times m$ or the $3 \times n$ graph will have either only the middle tile or only the top and bottom tiles of an end column covered, as in Figure \ref{fig:T_3 impossible tiling}, which is impossible to tile. This is because the only way to fill the top and bottom spaces of an end column missing its middle tile is with horizontal tiles, which then leaves the middle tile of the second column uncovered in a way that can only be covered by a horizontal tile. However, once that space has been covered, the graph is in the same state as it was at the beginning, though now two of its columns have been covered. This process repeats itself, forcing the placement of horizontal tiles until there is not enough space for another domino and either the middle vertex or the top and bottom vertices of the end column will always be uncovered, meaning there is no way to tile the whole graph. Therefore, these combinations of shifting and gluing contribute 0 to the total.\\\\

\begin{figure}[h]
    \centering
    \includegraphics[scale = 0.12]{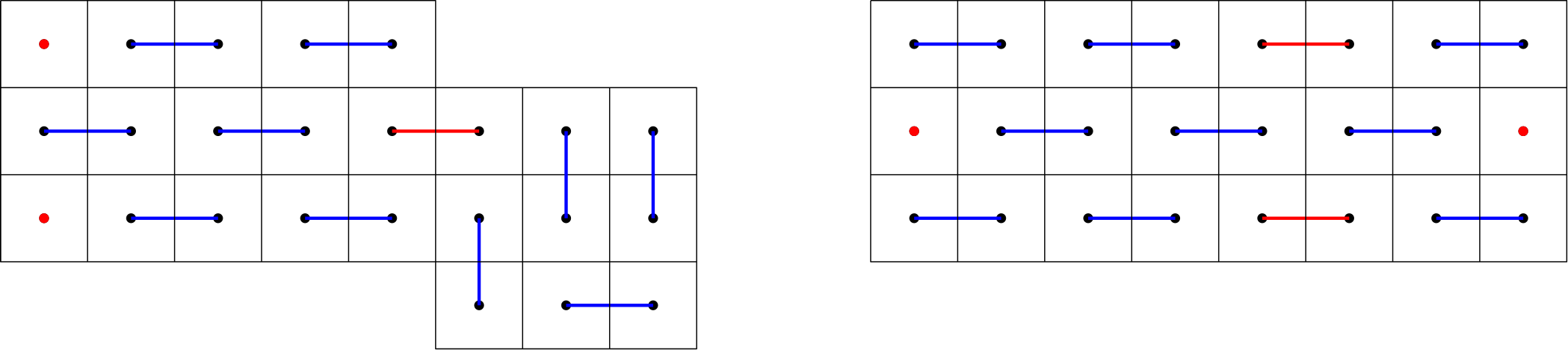}
    \caption{Impossible tiling configurations.}
    \label{fig:T_3 impossible tiling}
\end{figure}

When $s=1$ and $B=\{e_1, e_2\}$ or when $s=0$ and $B=\{e_1, e_2\}$ or $\{e_2, e_3\}$, then the graphs are positioned like in Figure \ref{fig:T_3 possible tiling}. If $m$ or $n$ is odd, then there will be an odd number of spaces left on either the $3 \times m$ or $3 \times n$ graph which is impossible to tile, so suppose both $m$ and $n$ are even. Then for each graph we have an odd number of rows with one extra tile, forcing us to place a domino horizontally next to the bridges. This leaves a $3 \times (m-1)$ graph missing a vertex to fill. In the proof of the $s=2$ and $B=\{e_1\}$ or when $s=0$ and $B=\{e_1\}$ or $\{e_3\}$ gluing cases we showed that the number of ways to finish tiling such a graph is given by $\frac{T_3(m)-T_3(m-2)}{2}$, so in this case the total number of ways to tile the graph is given by multiplying this result for $m$ with the same result for $n$. Adding these formulas together gives us the final result: $\frac{(T_3(m)-T_3(m-2))(T_3(n)-T_3(n-2))}{4}$.\\\\

\begin{figure}[h]
    \centering
    \includegraphics[scale = 0.12]{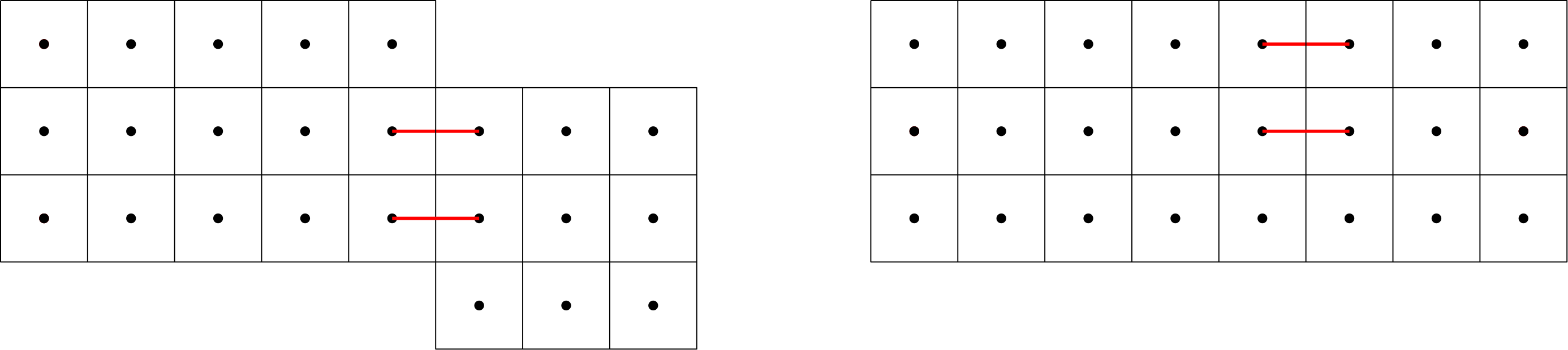}
    \caption{Initial tiling configurations of $L_{3,5} \sqcup_{1,\{e_1,e_2\}} L_{3,3}$ and $L_{3,5} \sqcup_{0,\{e_1,e_2\}} L_{3,3}$.}
    \label{fig:T_3 possible tiling}
\end{figure}

When $s=0$ and $B=\{e_1, e_2, e_3\}$, then the graphs are positioned like in Figure \ref{fig:G(3,3)}. This configuration leaves a $3 \times (m-1)$ and a $3 \times (n-1)$ graph to be tiled, so the number of ways to tile them is the number of ways to tile one of them multiplied by the number of ways to tile the other. That is, $T_3(m-1)T_3(n-1)$.

\begin{figure}[h]
    \centering
    \includegraphics[scale = 0.3]{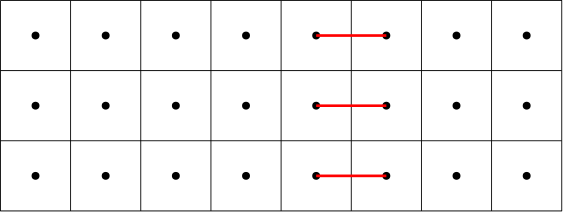}
    \caption{Initial tiling configuration of $L_{3,5} \sqcup_{0,\{e_1,e_2,e_3\}} L_{3,3}$.}
    \label{fig:G(3,3)}
\end{figure}

\end{proof}

In the proof for Theorem \ref{Theo:T10}, the cases for the number of perfect matchings were split based on the shift and bridges that defined the gluing. We can combine a number of these cases to express $T_3(m+n)$ in terms of tilings relating to $m$ and $n$ (See Corollary \ref{Theo:3B3}).

\subsection{The $4\times n$ case}
\begin{theorem}
The number of unique ways to tile $L_{4,n}$ with dominoes is given by the recursive formula $T_4(n)=T_4(n-1)+5T_4(n-2)+T_4(n-3)-T_4(n-4)$, where $T_4(1)=1$, $T_4(2)=5$, $T_4(3)=11$, and $T_4(4)=36$.
\end{theorem}
\begin{proof}
The initial values of the recursion can be checked by hand. As for the recursive formula, let $n$ be an arbitrary positive integer greater than 4. All possible tilings of $L_{4,n}$ must begin on the left-hand side with one of these initial configurations:\\\\
In the first case, where the first column is filled with vertical dominoes, there is a $4\times (n-1)$ graph left to tile, contributing $T_4(n-1)$ ways to fill the graph.\\
In the second case, where the first two columns are filled with horizontal dominoes, there is a $4\times (n-2)$ graph left to tile, contributing $T_4(n-2)$ ways to fill the graph.\\
In the third case, the two empty spaces in the second column can either be filled with a vertical tile or two horizontal tiles. Adding a vertical tile leaves a $4 \times (n-2)$ graph to tile and adding two horizontal tiles gives you the same choice of filling the two empty spaces in the third column either by adding a vertical tile or two horizontal tiles. You can continue choosing to add horizontal tiles until you run out of space on the graph, in total adding  $T_4(n-2)+T_4(n-3)+T_4(n-4)+...+T_4(0)=\sum_{i=0}^{n-2}T_4(i)$ ways to tile the graph. The same goes for the fourth case.\\
In the fifth case, the two empty spaces in the second column can either be filled with a vertical tile or two horizontal tiles. Adding a vertical tile leaves a $4 \times (n-2)$ graph to tile and adding two horizontal tiles forces you to place horizontal tiles in the top and bottom rows and then gives you the same choice of filling the two empty spaces in the fourth column either by adding a vertical tile or two horizontal tiles. You can continue choosing to add horizontal tiles until you run out of space on the graph, in total adding  \[T_4(n-2)+T_4(n-4)+T_4(n-6)+...+((T_4(1) \text{ if $n$ is odd) or } (T_4(0) \text{ if $n$ is even}))=\sum_{i=1}^{\lfloor\frac{n}{2}\rfloor} T_4(n-2i)\] ways to tile the graph.\\
So far, we have \[T_4(n)=T_4(n-1)+T_4(n-2)+2\sum_{i=0}^{n-2}T_4(i)+\sum_{i=1}^{\lfloor\frac{n}{2}\rfloor} T_4(n-2i).\] However, we can simplify this formula further:\\
\[T_4(n-2)=T_4(n-3)+T_4(n-4)+2\sum_{i=0}^{n-4}T_4(i)+\sum_{i=1}^{\lfloor\frac{n-2}{2}\rfloor} T_4(n-2-2i),\] so 
\begin{align*}
    T_4(n)-T_4(n-2)&=(T_4(n-1)+T_4(n-2)+2\sum_{i=0}^{n-2}T_4(i)+\sum_{i=1}^{\lfloor\frac{n}{2}\rfloor} T_4(n-2i))\\
    &-(T_4(n-3)+T_4(n-4)+2\sum_{i=0}^{n-4}T_4(i)+\sum_{i=1}^{\lfloor\frac{n-2}{2}\rfloor} T_4(n-2-2i))\\
    &=T_4(n-1)+4T_4(n-2)+T_4(n-3)-T_4(n-4)
\end{align*}
So $T_4(n)=T_4(n-1)+5T_4(n-2)+T_4(n-3)-T_4(n-4)$.
\end{proof}
In order to prove a gluing formula for the $4\times n$ case, it will be helpful to derive formulae for the sums in the proof of the recursion formula, see Propositions \ref{Prop:4consecutivesum} and \ref{Prop:4alternatingsum} in the Appendix.

We can also find information about gluing graphs together, considering different shifts and bridge sets. If the number of bridges is odd, then what it left is impossible to tile because an odd number of vertices will be left on each graph.
\begin{theorem}\label{Theo:4xnm}
The number of perfect matchings of $L_{4,m}$ and $L_{4,n}$ glued together with a shift $T_{4}(m \sqcup_{s,B} n)$ is given by 

\[
    \begin{cases}
    T_4(m)T_4(n) & B= \emptyset\\
    \frac{(T_4(m+1)-T_4(m-2))(T_4(n+1)-T_4(n-2))}{25} & s=2, B= \{e_1,e_2\}\\
    f(m)\frac{(T_4(n+1)-T_4(n-2))}{5} & s=1, B= \{e_1,e_2\}\\
    \frac{(T_4(m+1)-T_4(m-2))}{5}f(n) & s=1, B= \{e_2,e_3\}\\
    \frac{(T_4(m+1)-T_4(m-2))(T_4(n+1)-T_4(n-2))}{25} & s=0, B= \{e_1,e_2\} \text{ or } \{e_3,e_4\}\\
    0 & s=0 \text{ or } 1, B= \{e_1,e_3\} \text{ or } \{e_2,e_4\}\\ 
    f(m+1)f(n+1) & s=0, B= \{e_1,e_4\}\\
    f(m)f(n) & s=0, B= \{e_2,e_3\}\\
    T_4(m-1)T_4(n-1) & s=0, B= \{e_1,e_2,e_3,e_4\}\\

\end{cases}\\
\]
Where $f(n)=-\frac{2}{5}T_4(n)+4T_4(n-2)+\frac{7}{5}T_4(n-3)-T_4(n-4)$.
\end{theorem}

\begin{proof}
When $B=\emptyset$, there are no bridges between the two lattice graphs, so the number of ways to tile them is the number of ways to tile one of them multiplied by the number of ways to tile the other. That is, $T_4(m)T_4(n)$.\\\\

When $s=2$ and $B=\{e_1, e_2\}$ or when $s=0$ and $B=\{e_1, e_2\}$ or $\{e_3, e_4\}$, then the graphs are positioned like in Figure \ref{fig:s=2,0}. The method of solving will be the same for each graph in the gluing, so we will look at the $4 \times n$ graph. On the first row, there are 2 available vertices adjacent to one another, so we can split this into two cases as shown in \ref{fig:B2}: either include the vertical edge between the two, in which case there are $T_4(n-1)$ ways to finish the perfect matching, or we include the horizontal edges connecting them to the adjacent row. In this case we again have the choice between including a vertical edge, leaving us with a $4 \times (n-2)$ graph to tile, adding $T_4(n-2)$ ways to finish the perfect matching, or including two horizontal edges, again allowing for either a vertical tile or two horizontal tiles. This will continue until we have a $2 \times 1$ graph left, adding $1 = T_4(0)$ to the sum. Thus we get that the number of ways to tile $L_{4,n}$ missing one corner vertex and one vertex above or beneath it is equal to $T_4(n-1)+T_4(n-2)+...+T_4(0)$. By proposition \ref{Prop:4consecutivesum}, this sum equals $\frac{T_4(n+1)-T_4(n-2)}{5}$. Similarly, the number of ways to tile the left graph is $\frac{T_4(m+1)-T_4(m-2)}{5}$. Multiplying these two values together gives us the total number of tilings for these shift and bridge gluing combinations: $\frac{(T_4(m+1)-T_4(m-2))(T_4(n+1)-T_4(n-2))}{25}$.

\begin{figure}[h]
    \centering
    \includegraphics[scale = 0.15]{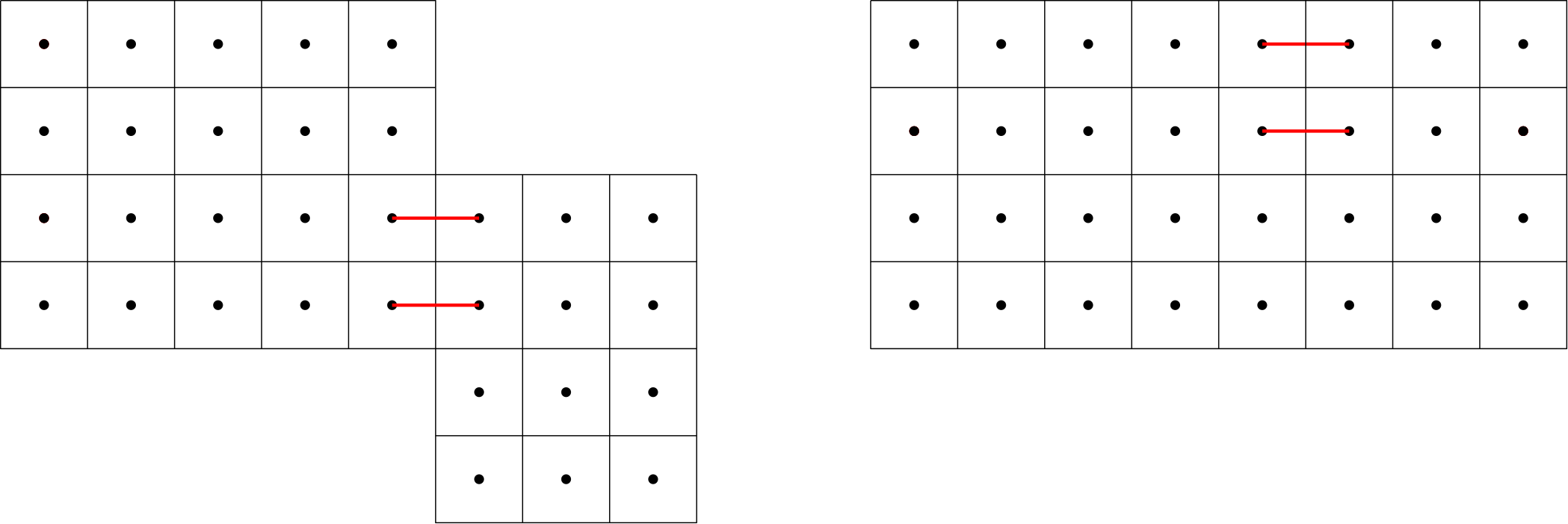}
    \caption{Initial tiling configurations of $L_{4,5} \sqcup_{2,\{e_1,e_2\}} L_{4,3}$ and $L_{4,5} \sqcup_{0,\{e_1,e_2\}} L_{4,3}$.}
    \label{fig:s=2,0}
\end{figure}

\begin{figure}[h]
    \centering
    \includegraphics[scale = 0.15]{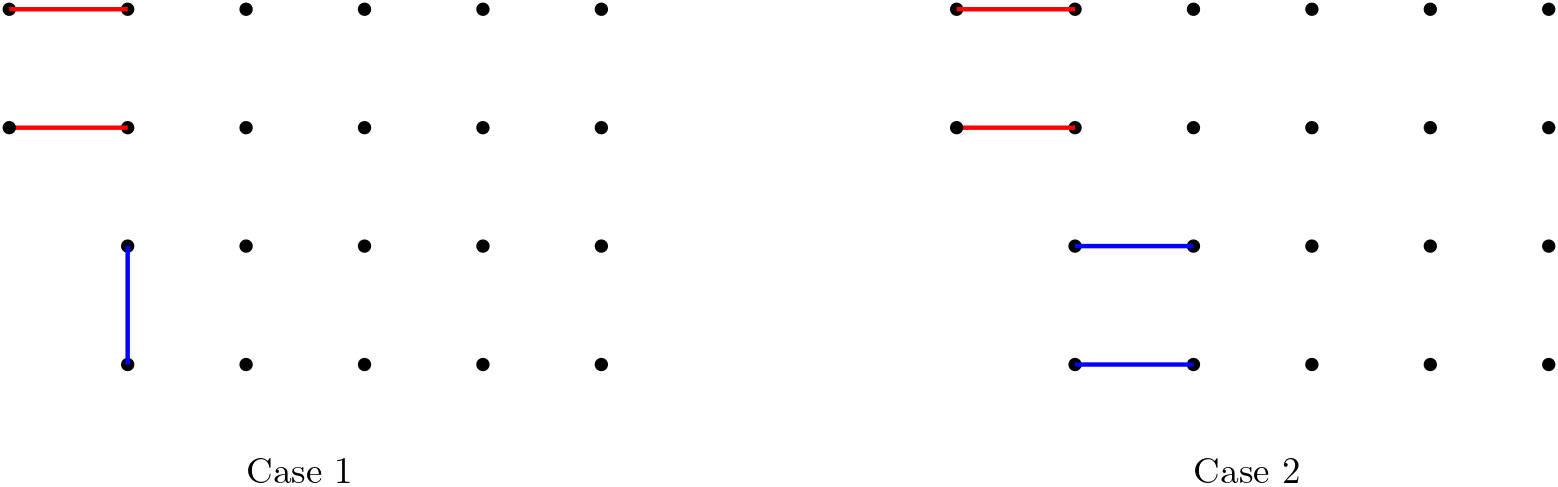}
    \caption{Intermediate tiling configurations of $L_{4,5}$ glued with 2 bridges.}
    \label{fig:B2}
\end{figure}

When $s=1$ and $B=\{e_1, e_2\}$ or $\{e_2, e_3\}$, then the graphs are positioned like in Figure \ref{fig:T_4,s=1}. This means that one of the graphs will be missing one corner vertex and another vertex above or beneath the corner and the other graph will be missing the middle two vertices from one of the end columns. We already have a formula from the previous paragraph for the number of perfect matchings of the first graph, so let us focus on the second. Since the middle two vertices of the end column are covered, we are forced to place two horizontal tiles on the corner spaces. To fill the second-to-last column, we can either use a vertical tile or two horizontal tiles, as shown in Figure \ref{fig:T_4,s=1}. In the case of a vertical tile, we are left with a $4 \times (m-2)$ graph, contributing $T_4(m-2)$ perfect matchings to the total. In the case of two horizontal tiles, we are again forced to place horizontal tiles in the top and bottom rows, now leaving the middle two vertices of the fourth column uncovered, and thus the pattern repeats, adding $T_4(m-2k)$ (where $1\leq k \leq \lfloor\frac{m}{2}\rfloor$) tilings in each iteration until there is either the middle two vertices of the last column or the entire last column and the middle two vertices of the second-to-last column left uncovered, depending on if $m$ is even or odd. Therefore, in total there are
\begin{eqnarray*}
&&T_4(m-2)+T_4(m-4)+T_4(m-6)+...+(T_4(1) \text{ if m is odd) or } (T_4(0) \text{ if m is even})\\
&=&\sum_{i=1}^{\lfloor\frac{m}{2}\rfloor} T_4(m-2i)=-\frac{2}{5}T_4(m)+4T_4(m-2)+\frac{7}{5}T_4(m-3)-T_4(m-4)\\
\end{eqnarray*}

ways to tile $L_{4,m}$ with the middle two vertices missing from the last column. From now on, we will abbreviate this quantity as $f(m)$. By multiplying this value by the number of ways to tile $L_{4,n}$ missing a corner vertex and a vertex directly below it, we see that there are \[f(m)\frac{(T_4(n+1)-T_4(n-2))}{5}\] ways to tile two lattice graphs shifted by 1 with a bridge set equal to $\{e_1, e_2\}$ and \[\frac{(T_4(m+1)-T_4(m-2))}{5}f(n)\] ways to tile two lattice graphs shifted by 1 with a bridge set equal to $\{e_2, e_3\}$.

\begin{figure}[h]
    \centering
    \includegraphics[scale = 0.15]{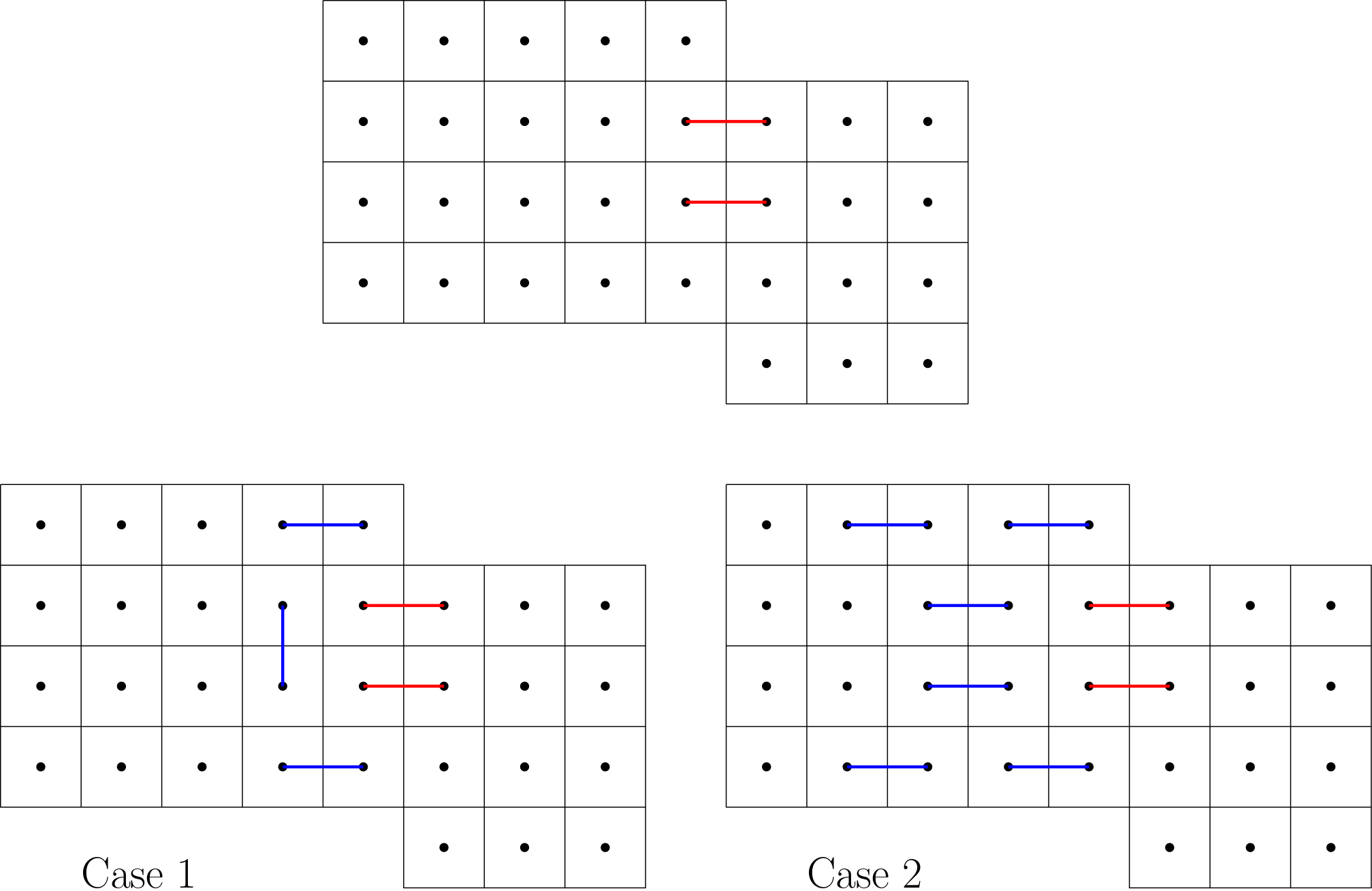}
    \caption{Initial and intermediate tiling configurations of $L_{4,5} \sqcup_{1,\{e_1,e_2\}} L_{4,3}$.}
    \label{fig:T_4,s=1}
\end{figure}

When $s=0$ or $1$ and $B=\{e_1, e_3\}$ or $\{e_2, e_4\}$, then at least one of the graphs has its corner vertex missing along with a non corner, nonadjacent vertex from the same column, which is impossible to tile. This is because the only way to cover the vertices of an end column missing a corner and a non corner, nonadjacent tile is with horizontal tiles, which then leaves two nonadjacent vertices of the second column uncovered in a way that can only be covered by two more horizontal tiles. However, once those vertices have been covered, the graph is in the same state as it was at the beginning, though now two of its columns have been covered. This process repeats itself, forcing the placement of horizontal tiles until there is no more space on the graph and either the middle vertex or the top and bottom vertices of the end column will always be uncovered, meaning there is no way to tile the whole graph. Therefore, these combinations of shifting and gluing contribute 0 to the total.

When $s=0$ and $B=\{e_1, e_4\}$, then we have two graphs next to each other, each missing two corner vertices from an end column, as shown in Figure \ref{fig:G(4,2)}. We will look at just the right graph, since the possible ways to tile each graph are the same. We can either use one vertical tile or two horizontal tiles to cover the vertices in the column missing its corners. If we use a vertical tile, there are $T_4(n-1)$ ways to tile the $4 \times (n-1)$ graph that remains. If we use two horizontal tiles, we are forced to place two more horizontal tiles in the top and bottom rows, so what's left is a $4 \times (n-2)$ graph with two of the corners missing from the left-most column. As we have seen before, we can keep reducing the length of the graph like this, adding $T_4(n-2k+1)$ (where $1\leq k \leq \lfloor\frac{n}{2}\rfloor$) perfect matchings to the total with each iteration until there is either the middle two vertices of the last column or the entire last column and the middle two vertices of the second-to-last column left uncovered, depending on if $n$ is even or odd. Therefore, in total there are $T_4(n-1)+T_4(n-3)+T_4(n-5)+...+(T_4(1) \text{ if n is even) or } (T_4(0) \text{ if n is odd})=\sum_{i=1}^{\lfloor\frac{n}{2}\rfloor} T_4(n-2i+1)$, which, by Proposition \ref{Prop:4alternatingsum} equals $f(n+1)$ ways to tile $L_{4,n}$ with the two corner vertices missing from the last column. The number of ways to tile the left graph is the same except for replacing every instance of $n$ with $m$. By multiplying these values together, we see that there are $f(n+1)f(m+1)$ ways to tile two lattice graphs shifted by 0 with a bridge set equal to $\{e_1, e_4\}$.

\begin{figure}[h]
    \centering
    \includegraphics[scale = 0.3]{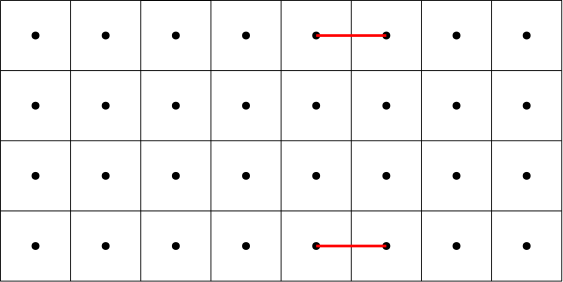}
    \caption{Initial tiling configuration of $L_{4,5} \sqcup_{0,\{e_1,e_4\}} L_{4,3}$.}
    \label{fig:G(4,2)}
\end{figure}

When $s=0$ and $B=\{e_2, e_3\}$, then we have two graphs next to each other, each missing the two middle vertices from their end columns. From a previous paragraph, we know there are \[-\frac{2}{5}T_4(m)+4T_4(m-2)+\frac{7}{5}T_4(m-3)-T_4(m-4)\] ways to tile such graphs. By multiplying these values together, we see that there are \[f(m)f(n)\] ways to tile two lattice graphs shifted by 0 with a bridge set equal to $\{e_2, e_3\}$.\\\\

When $s=0$ and $B=\{e_1, e_2, e_3, e_4\}$, then the graphs are positioned like in Figure \ref{fig:G(4,4)}. This configuration leaves a $4 \times (m-1)$ and a $4 \times (n-1)$ graph to be tiled, so the number of ways to tile them is the number of ways to tile one of them multiplied by the number of ways to tile the other. That is, $T_4(m-1)T_4(n-1)$.

\begin{figure}[h]
    \centering
    \includegraphics[scale = 0.3]{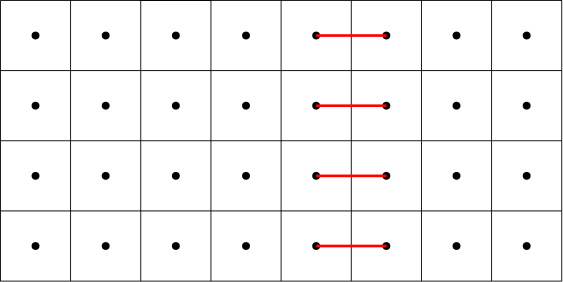}
    \caption{Initial tiling configuration of $L_{4,5} \sqcup_{1,\{e_1,e_2,e_3,e_4\}} L_{4,3}$.}
    \label{fig:G(4,4)}
\end{figure}
\end{proof}

\section{Spinors, Clifford algebras on graphs and gluing}

Spinors can be examined as elements of linear representation of Clifford Algebras. For example, the space of Pauli spinors can be approached as the space $Cl^-(0,3)$ and the space of Dirac spinors as $Cl^+(1,3)$ \cite{renaud:hal-03015551}. The relations between the generators of the algebra parallel the rules governing interactions between spinors in a quantum system. While there are a number of candidate Clifford algebras that have been explored as representations of the space of spinors \cite{renaud:hal-03015551}, we focus on a special type of Clifford algebra associated to a graph, introduced by Khovanova \cite{Khovanova}.

\begin{definition}
Let $\Gamma$ be a graph with $n$ vertices. Its associated Clifford algebra $A_\Gamma$ has $n$ generators $e_1,\ldots,e_n$ corresponding to each vertex. For each $i$, $e_i^2 = -1$ and 

$\begin{cases}
    e_i e_j = -e_j e_i & i \text{ and } j \text{ are adjacent in } \Gamma\\
    e_i e_j = e_i e_j & i \text{ and } j \text{ not adjacent in } \Gamma
\end{cases}$.
\end{definition}

The center of a Clifford graph algebra characterizes the structure of the whole algebra. Each central monomial of a Clifford graph algebra gives a decomposition as a direct sum of two algebras. As a result, the structure of a Clifford graph algebra can be identified solely by the number of vertices and the dimension of its center \cite{Khovanova}.

The center of a Clifford graph algebra is spanned by its monomials. They are determined by the structure of the graph \cite{Khovanova}, as stated in the following

\begin{lemma}
A monomial $e_\alpha$ is central if and only if for each vertex $i \in \Gamma$, there are an even number of edges connecting $i$ to $\alpha$. 
\end{lemma}

\subsection{Examples}

We give some examples of Clifford graph algebras and their centers.

\begin{proposition}\label{Prop:CenterPath}
Let $P_n$ be a path graph. Then 
$$Z(A_{P_n}) = 
\begin{cases}
    \mathbb{C} & \text{ if } n \text{ is even}\\
    \mathbb{C}^2 & \text{ if } n \text{ is odd}
\end{cases}.$$

\end{proposition}
To be precise, for odd $n$, $Z(A_{P_n})$ is spanned by the central monomial $e_1 e_3 \cdots e_n$ and $1$.

Consider the path graph $P_5$, labeled as follows:
\begin{figure}[h]
\includegraphics[scale=0.1]{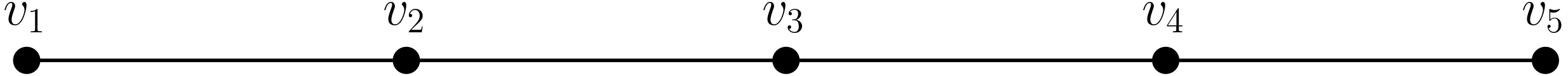}
\end{figure}

The center of $P_7$ is $Z(A_{P_5}) = \{ \alpha \, e_1 e_3 e_5  + \beta \; | \; \alpha, \beta \in \mathbb{C} \}$.

$e_1 e_3 e_5$ is a central monomial in $A_{P_7}$ because we can check that each vertex is connected to $\alpha = \{ 1,3,5 \}$ by an even number of edges.

All edges in $P_5$ connect an even vertex to an odd vertex. This means for odd $i$, there are no edges connecting it to $\alpha$. For even $i$, the only incident edges are $(i,i-1)$ and $(i,i+1)$. Both $i-1$ and $i+1$ are in $\alpha$, so there are two edges connecting it to $\alpha$.

\subsection{Graph gluing and Clifford algebras}


\begin{definition}
Let $V$, $W$ be two $\mathbb{C}$-algebras with bases $\{ v_i \}_{i=1}^n$ and $\{ w_j \}_{j=1}^n$. The tensor product $V \otimes W$ has basis $\{ v_i \otimes w_j \}$ such that

\begin{enumerate}
    \item $(v_1 + v_2) \otimes w = v_1 \otimes w + v_2 \otimes w$
    \item $v \otimes (w_1 + w_2) = v \otimes w_1 + v \otimes w_2$
    \item $\lambda(v \otimes w) = (\lambda v) \otimes = v \otimes (\lambda w), \lambda \in \mathbb C$.
\end{enumerate}
\end{definition}

In the case of Clifford graph algebras: $V$ and $W$ are $\mathbb{C}$-algebras spanned by monomials. Assume the generators $\{e_k\}$ for $V$ are distinct from the generators $\{e_l\}$ for $W$. Each element $v \in V$ and $w \in W$ is a polynomial:

\begin{enumerate}
    \item $v = \sum\limits_{i=1}^n \alpha_i v_i$
    \item $w = \sum\limits_{j=1}^m \lambda_j w_j,$
\end{enumerate}
where $\alpha_i, \lambda_j \in \mathbb C$.
Following the normal rules of multiplying two monomials, we define the tensor product $v_i \otimes w_j = v_i w_j$. Because the monomial $v_i$ is defined in different generators from $w_j$, $v_i w_j$ is a monomial of degree 1.

Under the normal polynomial multiplication, we have:

\[ v \otimes w = vw = \sum_{i=1}^n \sum_{j=1}^m \alpha_i \lambda_j v_i w_j.\]

Every element $v \otimes w \in V \otimes W$ is a linear combination of the monomials $v_i \otimes w_j = v_i w_j$.

In addition, polynomial multiplication satisfies the following properties:
\begin{enumerate}
    \item $(v_1 + v_2)w = v_1 w + v_2 w \iff (v_1 + v_2) \otimes w = v_1 \otimes w + v_2 \otimes w$
    \item $v (w_1 + w_2) = v w_1 + v w_2 \iff v \otimes (w_1 + w_2) = v \otimes w_1 + v \otimes w_2$
    \item $\lambda vw = (\lambda v) w = v(\lambda w) \iff \lambda(v \otimes w) = (\lambda v) \otimes = v \otimes (\lambda w)$.
\end{enumerate}

\begin{theorem}\label{thm: Clifford_Gluing_Center}
Let $\Gamma_1$ and $\Gamma_2$ be graphs and $A_{\Gamma_1}$ and $A_{\Gamma_2}$ be their associated Clifford algebras. Then $Z(A_{\Gamma_1}) \otimes Z(A_{\Gamma_2}) = Z(A_{\Gamma_1 \sqcup \Gamma_2}).$
\end{theorem}
\begin{proof}
($\Rightarrow$) First, we need to show $Z(A_{\Gamma_1}) \otimes Z(A_{\Gamma_2}) \subseteq Z(A_{\Gamma_1 \sqcup \Gamma_2})$.

Let $\{ e_\alpha \}$ be the basis of central monomials for $Z(A_{\Gamma_1})$ and $\{ e_\beta \}$ be the basis of central monomials for $Z(A_{\Gamma_2})$.

Then $\{ e_\alpha e_\beta \}$ is a basis for $Z(A_{\Gamma_1 \sqcup \Gamma_2})$. We want to show it is contained in $Z(A_{\Gamma_1} \sqcup A_{\Gamma_2})$.

Let $e_\alpha e_\beta$ be fixed. Observe $e_\alpha e_\beta$ is a monomial associated with the vertex set $\alpha \sqcup \beta$, so we can denote $e_{\alpha \beta} = e_\alpha e_\beta$.

Recall a monomial $e_\alpha$ is central in $A_\Gamma$ if and only if for all $v \in V(\Gamma)$, the number of edges connecting it to $\alpha$ is even.

In order to show $e_{\alpha \beta} \in Z(A_{\Gamma_1 \sqcup \Gamma_2})$, we need to show for every $v \in V(\Gamma_1 \sqcup \Gamma_2)$, there are an even number of edges connecting it to $\alpha \sqcup \beta$.
\newline

Case 1: $v \in V(\Gamma_1)$. Because $e_\alpha$ is a central monomial of $A_{\Gamma_1}$, there are an even number of edges connected $v$ to $\alpha$. $\Gamma_1$ and $\Gamma_2$ are disconnected in $\Gamma_1 \sqcup \Gamma_2$, so there are no edges between $v$ and $\beta$.

Together, there are an even number of edges connecting $v$ to $\alpha \sqcup \beta$.
\newline

Case 2: $v \in V(\Gamma_2)$. Because $e_\beta$ is a central monomial of $A_{\Gamma_2}$, there are an even number of edges connected $v$ to $\beta$.

$\Gamma_1$ and $\Gamma_2$ are disconnected in $\Gamma_1 \sqcup \Gamma_2$, so there are no edges between $v$ and $\alpha$.
\newline

Together, there are an even number of edges connecting $v$ to $\alpha \sqcup \beta$.

We thus conclude that $e_\alpha e_\beta = e_{\alpha \beta}$ is a central monomial of $A_{\Gamma_1 \sqcup \Gamma_2}$. It follows that $Z(A_{\Gamma_1}) \otimes Z(A_{\Gamma_2}) \subseteq Z(A_{\Gamma_1 \sqcup \Gamma_2})$.
\newline

($\Leftarrow$) Next, we need to show $Z(A_{\Gamma_1 \sqcup \Gamma_2}) \subseteq Z(A_{\Gamma_1}) \otimes Z(A_{\Gamma_2})$.

Let $\{ e_\alpha \}$ be the basis of monomials for $Z(A_{\Gamma_1 \sqcup \Gamma_2})$. We want to show it is contained in $Z(A_{\Gamma_1}) \otimes Z(A_{\Gamma_2})$.

Let $e_\alpha$ be a central monomial of $A_{\Gamma_1 \sqcup \Gamma_2}$. Because $\Gamma_1$ and $\Gamma_2$ are disconnected components of $\Gamma_1 \sqcup \Gamma_2$, we can separate the vertex set $\alpha = \alpha_1 \sqcup \alpha_2$, where $\alpha_1$ and $\alpha_2$ are the subsets of vertices from $\Gamma_1$ and $\Gamma_2$, respectively.

We want to show $e_\alpha = e_{\alpha_1} e_{\alpha_2} \in Z(A_{\Gamma_1}) \otimes Z(A_{\Gamma_2})$, or $e_{\alpha_1} \in Z(A_{\Gamma_1})$ and $e_{\alpha_2} \in Z(A_{\Gamma_2})$. 

Start with $e_{\alpha_1}$. Let $v \in V(\Gamma_1)$. 

Because $e_\alpha$ is a central monomial of $\Gamma_1 \sqcup \Gamma_2$, we know there are an even number of edges between $v$ and $\alpha_2 \sqcup \alpha_2$. In particular, there are no edges between $v$ and $\alpha_2$ because $\Gamma_1$ and $\Gamma_2$ are disconnected. Then the number of edges connecting $v$ to $\alpha_1$ is the same for $\alpha_1 \sqcup \alpha_2$, which is even. We have shown for any $v \in V(G_1)$, there are an even number of edges connecting it to $\alpha_1$, so $e_{\alpha_1}$ is a central monomial of $A_{\Gamma_1}$. Similarly, we can find that $e_{\alpha_2}$ is a central monomial of $A_{\Gamma_2}$. Then $e_\alpha = e_{\alpha_1} e_{\alpha_2} \in Z(A_{\Gamma_1}) \otimes Z(A_{\Gamma_2})$. We conclude $Z(A_{\Gamma_1 \sqcup \Gamma_2}) \subseteq Z(A_{\Gamma_1}) \otimes Z(A_{\Gamma_2})$, as desired.
\end{proof}

The centers of Clifford graph algebras from bridge gluings are more challenging to identify, but we can start with some simple examples.

There are two ways to glue the path graphs $P_2$ and $P_3$ with one bridge:

\begin{figure}[h]
    \centering
    \includegraphics[scale=0.2]{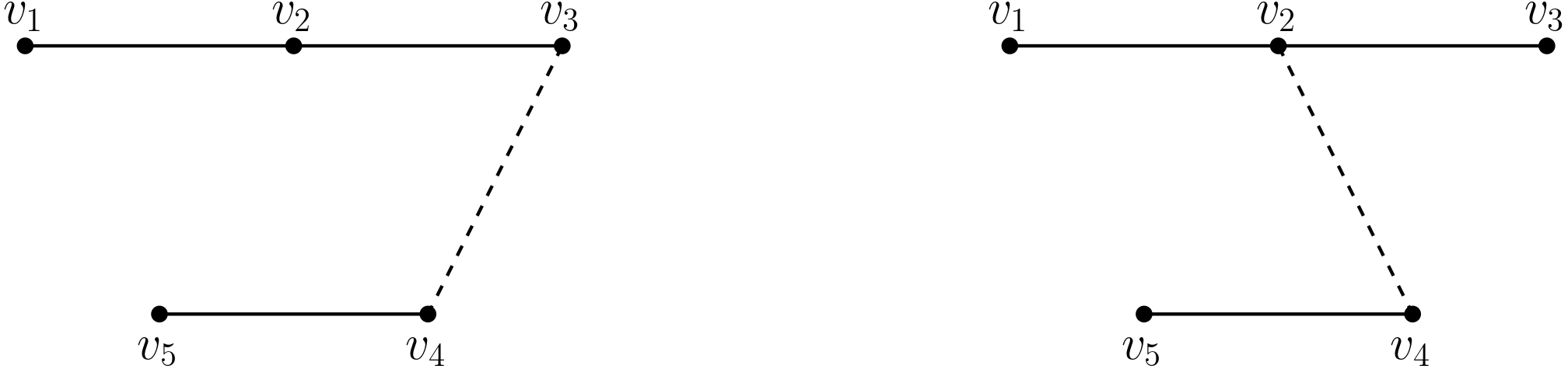}
        \caption{Two different bridge gluings between $P_3$ and $P_2$.}
\end{figure}

By Proposition \ref{Prop:CenterPath}, the bases for $Z(A_{P_2}$ and $Z(A_{P_3})$ are $\{ 1 \}$ and $\{ 1, e_1 e_3 \}$, respectively. In the left gluing, $P_2 \sqcup_{B_1} P_3 = P_5$, which has a center spanned by $e_1 e_3 e_5$ and $1$. The right gluing also has basis $\{ e_1 e_3 e_5, 1 \}$. With the same number of vertices and dimension of their centers, the Clifford algebras associated to the two bridge gluings are isomorphic.

This is not always the case though. Take the following gluing of two copies of $P_3$:

\begin{figure}[h]
    \centering
    \includegraphics[scale=0.4]{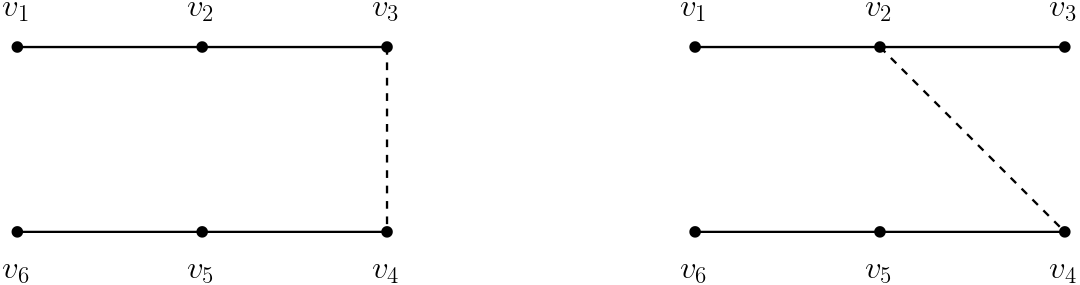}
    \caption{Two different bridge gluings between two copies of $P_3$.}
\end{figure}

The Clifford algebra of the left gluing has a center of dimension one, while the center for the right gluing has dimension four with central monomials $e_1 e_3$, $e_1 e_4 e_6$, and $e_3 e_4 e_6$.

We can characterize the dimension of the Clifford algebras of glued path graph. The case of bridge gluing the endpoints of two path graphs is simple: by Proposition \ref{Prop:CenterPath}, $P_n \sqcup_{B_1} P_m$ has a center of dimension $1$ if $n$ and $m$ have the same parity and $0$ if they have opposite parity.

\begin{theorem}\label{thm: Clifford_Gluing_Path}
Let $\Gamma = P_n \sqcup_{B_1} P_m$ be two path graphs glued by a bridge between an interior vertex and endpoint.
\begin{enumerate}
    \item If $n$ and $m$ are both even, then $Z(A_\Gamma)$ has dimension 1.
    \item If $n$ is even and $m$ is odd, then $Z(A_\Gamma)$ has dimension 2.
    \item If $n$ and $m$ are both odd, then $Z(A_\Gamma)$ has dimension 4.
\end{enumerate}
\end{theorem}

We prove the following lemma to identify central monomials of the glued graph Clifford algebra:

\begin{lemma}
Let $\Gamma$ be a tree. Suppose $e_\alpha$ is a central monomial of $A_\Gamma$. Then 
\begin{enumerate}
    \item there are at least one vertex of degree one in $\alpha$,
    \item for every vertex $v \in \alpha$, there exists $u \in \alpha$ such that $d(u,v) = 2$, and
    \item no pair of vertices $u,v \in \alpha$ are adjacent.
\end{enumerate}
\end{lemma}

This will help us to prove Theorem \ref{thm: Clifford_Gluing_Path}, because we only need to check the vertex sets that start at a leaf and skip every other vertex.

\begin{proof}
Let $e_\alpha$ be a central monomial of $A_\Gamma$.

\begin{enumerate}
    \item Suppose for the sake of contradiction, there is no vertex of degree one in $\alpha$.

Pick any interior vertex of $\Gamma$ as the root $r$. This gives $\Gamma$ a partial ordering relation $<_r$, where $x <_r y$ if $x$ lies in the $r-y$ path. Pick any maximal element $x \in \alpha$, that is, there is no $y \in \alpha$ such that $x <_r y$. By assumption, $x$ is not a leaf, so there is some vertex $v$ that is adjacent to $x$ away from the root $r$. Then $v$ is connected to $\alpha$ by a single edge $(v,x)$. There are no other edges between $v$ and $\alpha$ because for any other neighbor $u$, $x <_r u$, so $u \notin \alpha$ by the maximality of $x$.
We conclude $e_\alpha$ is not a central monomial, in contradiction.

    \item Let $v \in \alpha$. Suppose for the sake of contraction, there is no vertex $u \in \alpha$ such that $d(u,v)=2$. Pick any vertex $x$ adjacent to $v$. Then $x$ is connected to $\alpha$ by a single edge $(v,x)$. There are no other edges between $x$ and $\alpha$ because for any other neighbor $y$, $d(y,v)=2$. We conclude $e_\alpha$ is not a central monomial, in contradiction.
    
    \item Let $v \in \alpha$. Suppose for sake of contradiction, $u,v \in \alpha$ that are adjacent.

Case 1: $u$ or $v$ is a leaf. Assume WLOG $v$ is a leaf. Then $v$ is connected to $\alpha$ by a single edge $(u,v)$, since $v$ has no other neighbors.

Case 2: $u$ and $v$ are both interior vertices. At least one neighbor of $v$, say $x$, is in $\alpha$, otherwise $v$ only has one edge connecting it to $\alpha$. We can show the same for $x$, at least one of its neighbors is in $\alpha$. We repeat until we have found some leaf $z$ in $\alpha$, which is connected to $\alpha$ by a single edge. In contradiction, $e_\alpha$ is not a central monomial.
    
\end{enumerate}
\end{proof}

One corollary of this proof is that if a central vertex set starts at one leaf, then branches out, skipping every other vertex, then it must contain a leaf at the end of the branch. Otherwise, like in part 1, the ``next" vertex in the branch connects to the vertex set by one edge. As a result, the vertex set corresponding to a central monomial must contain at least two leaves that are an even distance apart.

Now we are able to prove Theorem \ref{thm: Clifford_Gluing_Path}.

\begin{proof}
 Let $P_n$ and $P_m$ be two path graphs. Label the vertices of $P_n$, $v_1,\ldots,v_n$, and similarly for $P_m$. Let $\Gamma$ be $P_n$ and $P_m$ glued by bridge $(v_k,u_1)$, where $v_k$ is an internal vertex of $P_n$ and $u_1$ is an endpoint of $P_m$.
 \newline

Case 1: $n$ and $m$ are both even.

Because $n$ is even, we know $v_k$ is an even distance from some endpoint of $P_n$ and an odd distance from the other. Assume WLOG $v_k$ is an even distance from $v_1$. The only pair of leaves that are an even apart is $v_1$ and $u_1$, so any central vertex set must contain $v_1,\ldots,v_k,u_2,\ldots,u_m$. We note $v_{k-1}$ connects to the set by a single edge $(v_{k-1},v_k)$, so its associated monomial is not central. However, if it branches towards $v_{k-2}$ from $v_k$, it will not contain leaf $v_1$ because they are odd distances apart.

\begin{figure}[h]
    \centering
    \includegraphics[scale=0.7]{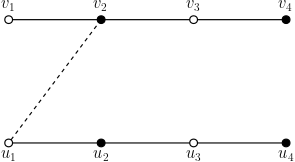}
    \caption{A candidate vertex set}
\end{figure}

Case 2: $n$ is even and $m$ is odd.

Let's assume without loss of generality that $v_k$ is at an even distance from $v_1$. We observe $v_1,\ldots,v_k,u_2,\ldots,v_m$ corresponds to a central monomial. Moreover, this is the only central monomial because the other pairs of leaves are an odd distance apart.
\begin{figure}[H]
    \centering
    \includegraphics[scale=0.08]{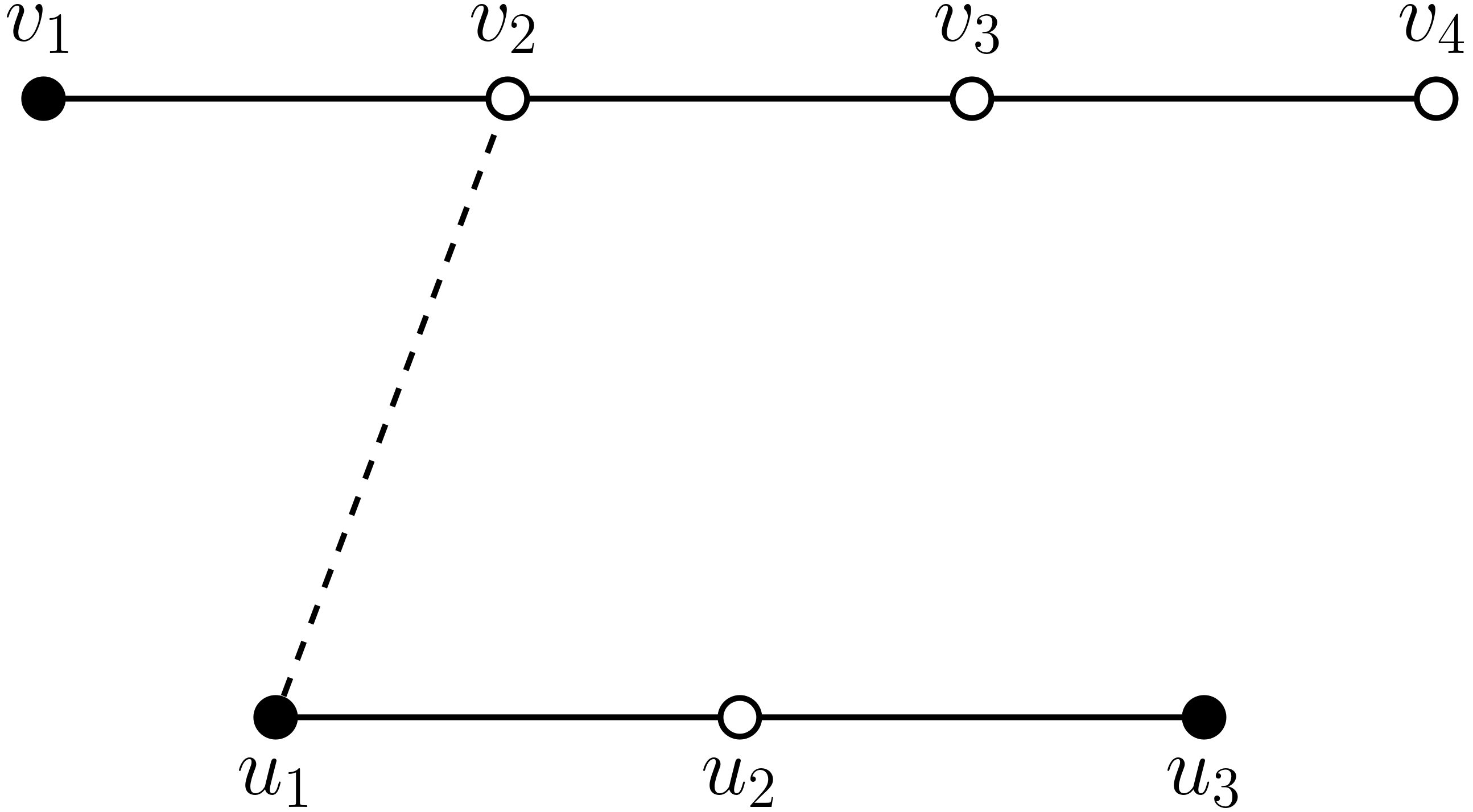}
    \caption{The vertex set associated to the unique central monomial of $A_\Gamma$}
    \label{fig:fig20}
\end{figure}

Case 3: $n$ is odd and $m$ is even.

Because $n$ is odd, $v_{nk}$ is either an odd distance from both endpoints or an even distance from both endpoints of $P_n$.

Suppose $v_{nk}$ is an even distance from both $v_{n1}$ and $v_{nn}$. Then $v_{n1},\ldots,v_{nk},\ldots,v_{nn},v_{m2},\ldots,v_{mm}$.

If $v_{nk}$ is an odd distance from both $v_{n1}$ and $v_{nn}$, then it is equivalent to Case 2. $\Gamma$ is also two paths, $v_1,\ldots,v_k,u_1,\ldots,u_m$ and $v_{k+1},\ldots,v_n$, glued by bridge $(v_k,v_{k+1})$. The former has an even number of vertices while the latter has an odd number, like in Case 2.
\newline

Case 4: $n$ and $m$ both odd.

Assume $k$ is even. Otherwise, if $k$ is odd, $\Gamma$ can also be seen as two paths, $v_{n1},\ldots, v_{nk},v_{m1},\ldots,v_{mm}$ and $v_{n(k+1)},\ldots,v_{nn}$ glued by bridge $(v_{nk},v_{n(k+1)})$. Both path graphs have even length, as in Case 1.

The vertex sets associated to central monomials are:
\begin{itemize}
    \item $v_1,\ldots,v_{k-1},v_{k+1},\ldots,v_n$
    \item $v_1,\ldots,v_{k-1},u_1,\ldots,u_m$
    \item $v_n,\ldots,v_{k+1},u_1,\ldots,u_m$.
\end{itemize}
\end{proof}

\section{Future directions}

Currently, our program to compute and create visualizations of solutions to the Dirac equation over time has been used to numerically verify our results. Future directions include investigating whether this process can be carried out in reverse, i.e. can information about the underlying graph or initial quantum state be extracted given these graphs and data of solutions to the Dirac equation over time.

Our gluing formulae help us to understand the behavior of spinors and Dirac operators on graphs as they become more complex. 
A formula that generalizes the results in Section \ref{sec:Dimer} is still open, and its complexity is expected to increase with the size of the lattice graphs. We also hope to generate gluing formulae for other types of graphs, e.g. honeycomb graphs, and planar graphs in general.

We have discovered explicit formulae for the quadratic forms of the odd Laplacian and the incidence Dirac operator. Future work could be carried out to determine explicit formulae for the quadratic forms of the even and odd Dirac operators.

The case of Clifford graph algebras for disjoint gluings and bridge gluings of path graphs help to understand the space of spinors in higher dimensions. Further research in generalizing bridge gluings will be foundational in understanding the Clifford algebras of more complex graphs. 

While our results on gluing formalae focus on the Clifford graph algebras defined by Khovanova \cite{Khovanova}, the Clifford algebras that represent spaces of spinors extend beyond those introduced. For example, $Cl(1,3)^+$, the even part of the Clifford algebra $Cl(1,3)$ is suitable for describing Dirac spinors.\cite{renaud:hal-03015551} $Cl(1,3)$ has four generators: one of which squares to $+1$, while the remaining three square to $-1$. In order to allow for such a distinction under Khovanov's framework of Clifford graph algebras, we propose an additional graph coloring. We propose a form of Clifford graph algebra where vertices $i$ are given a color based on their associated generator $e_i$, for example, red, if $e_i^2 = +1$, or blue, if $e_i^2 = -1$. Because the generators of Clifford algebras anticommute, their colored graphs are complete. 

\appendix
\section*{Appendix: Algebraic computations for the gluing of lattice graphs}
\renewcommand{\thesection}{A} 
\begin{corollary}\label{Theo:3B3}
The number of perfect matchings of $L_{3,n+m}$ is given by 
\begin{multline*}
T_3(m+n)  = T_3(m)T_3(n) + T_3(m-1) T_3(n-1) \\
+ \frac{(T_3(m+1) - T_3(m-1))(T_3(n+1)-T_3(m-1))}{2} \\
+ \frac{(T_3(m) - T_3(m-2))(T_3(n)-T_3(n-2))}{2}.
\end{multline*}  
\end{corollary}
\begin{proof}
 \begin{align*}
     T_3(m+n) &= \sum_B T_{3}(m \sqcup_{0,B} n)\\
     &= T_3(m)T_3(n) + 2\frac{(T_3(m+1) - T_3(m-1))(T_3(n+1)-T_3(m-1))}{4}\\
     &+ 2(0) + 2\frac{(T_3(m) - T_3(m-2))(T_3(n)-T_3(n-2))}{4}\\
     &+ T_3(m-1)T_3(n-1).
 \end{align*}
\end{proof}
\begin{proposition}\label{Prop:T3explicit}
The number of ways to tile $L_{3,n}$ with dominoes is given by the explicit formula $$\frac{(1-(-1)^{n+1})(\alpha^{n+1}+(\frac{1}{\alpha})^{n+1})}{2\sqrt{6}},$$ where $\alpha=\frac{\sqrt{2}+\sqrt{6}}{2}$.
\end{proposition}
\begin{proof}
We can represent the recursion equation using matrices: \begin{align*}
    \begin{bmatrix}
    T_3(n+3) \\
    T_3(n+2) \\
    T_3(n+1) \\
    T_3(n)
    \end{bmatrix}
    &= \begin{bmatrix}
    0 & 4 & 0 & -1 \\
    1 & 0 & 0 & 0 \\
    0 & 1 & 0 & 0 \\
    0 & 0 & 1 & 0
    \end{bmatrix}^{n}
    \begin{bmatrix}
    T_3(3) \\
    T_3(2) \\
    T_3(1) \\
    T_3(0)
    \end{bmatrix}.\\
    \end{align*}\\
By diagonalizing the recursion matrix and performing matrix multiplication, we find that $$T_3(n)=\frac{3-\sqrt{3}}{12}(\frac{-\sqrt{2}+\sqrt{6}}{2})^n+\frac{3-\sqrt{3}}{12}(\frac{\sqrt{2}-\sqrt{6}}{2})^n+\frac{3+\sqrt{3}}{12}(\frac{\sqrt{2}+\sqrt{6}}{2})^n+\frac{3+\sqrt{3}}{12}(\frac{-\sqrt{2}-\sqrt{6}}{2})^n.$$
By rearranging, the formula becomes: $$T_3(n)=\frac{\alpha^{n+1}-(\frac{-1}{\alpha})^{n+1}}{2\sqrt{6}}+\frac{(\frac{1}{\alpha})^{n+1}-(-\alpha)^{n+1}}{2\sqrt{6}} = (1-(-1)^{n+1})\frac{\alpha^{n+1}+(\frac{1}{\alpha})^{n+1}}{2\sqrt{6}}$$
where $\alpha=\frac{\sqrt{2}+\sqrt{6}}{2}$.
\end{proof}

\begin{proposition}\label{Prop:4consecutivesum}
The following identity holds:
\[\sum_{i=0}^{n-2}T_4(i)=\frac{T_4(n)-T_4(n-3)}{5}.\]
\end{proposition}
\begin{proof}
$T_4(n)=T_4(n-1)+T_4(n-2)+2\sum_{i=0}^{n-2}T_4(i)+T_4(n-2)+T_4(n-4)+T_4(n-6)+...+(T_4(1) \text{ if n is odd) or } (T_4(0) \text{ if n is even})$. By substituting $n-1$ for $n$, we see that $T_4(n-1)=T_4(n-2)+T_4(n-3)+2\sum_{i=0}^{n-3}T_4(i)+T_4(n-3)+T_4(n-5)+T_4(n-7)+...+(T_4(0) \text{ if n is odd) or } (T_4(1)+T_4(-1) \text{ if n is even})$. Since $T_4(-1)=0$, we can ignore that term. By adding these two equations together, we find that \begin{align*}
T_4(n)+T_4(n-1)&=(T_4(n-1)+T_4(n-2)+2\sum_{i=0}^{n-2}T_4(i)\\&+T_4(n-2)+T_4(n-4)+T_4(n-6)+...\\&+(T_4(1) \text{ if n is odd) or } (T_4(0) \text{ if n is even}))\\&+(T_4(n-2)+T_4(n-3)+2\sum_{i=0}^{n-3}T_4(i)\\&+T_4(n-3)+T_4(n-5)+T_4(n-7)+...\\&+(T_4(0) \text{ if n is odd) or } (T_4(1)+T_4(-1) \text{ if n is even}))\\&=T_4(n-1)+T_4(n-3)+5\sum_{i=0}^{n-2}T_4(i).  
\end{align*}
From this equality, we can deduce that \[\sum_{i=0}^{n-2}T_4(i)=\frac{T_4(n)-T_4(n-3)}{5}.\]
\end{proof}
\begin{proposition}\label{Prop:4alternatingsum}
The following identity holds:
\[\sum_{i=1}^{\lfloor\frac{n}{2}\rfloor} T_4(n-2i)=-\frac{2}{5}T_4(n)+4T_4(n-2)+\frac{7}{5}T_4(n-3)-T_4(n-4).\]
\end{proposition}
\begin{proof}
Again, we start with $T_4(n)=T_4(n-1)+T_4(n-2)+2\sum_{i=0}^{n-2}T_4(i)+\sum_{i=1}^{\lfloor\frac{n}{2}\rfloor} T_4(n-2i)$. From the previous proposition, we can rewrite this formula as $T_4(n)=T_4(n-1)+T_4(n-2)+2\frac{T_4(n)-T_4(n-3)}{5}+\sum_{i=1}^{\lfloor\frac{n}{2}\rfloor} T_4(n-2i)$
We also know that $T_4(n)=T_4(n-1)+5T_4(n-2)+T_4(n-3)-T_4(n-4)$. The result follows from equating these two versions of $T_4(n)$ and solving for $\sum_{i=1}^{\lfloor\frac{n}{2}\rfloor} T_4(n-2i)$.
\end{proof}

\begin{corollary}\label{Theo:4B4}
The number of perfect matchings of $L_{4,n+m}$ is given by 
\begin{align*}
T_4(m+n)  &= T_4(m)T_4(n) + T_4(m-1)T_4(n-1) \\
&+ \frac{2(T_4(m+1) - T_4(m-2))(T_4(n+1)-T_4(n-2))}{25} \\
&+ f(m+1)f(n+1) \\
&+ f(m)f(n).  
\end{align*}
Where $f(n)=-\frac{2}{5}T_4(n)+4T_4(n-2)+\frac{7}{5}T_4(n-3)-T_4(n-4)$.
\end{corollary}

\begin{proof}
 \begin{align*}
    T_4(m+n) &= \sum_B T_{4}(m \sqcup_{0,B} n)\\
    &= T_4(m)T_4(n) + 4(0) + 2\frac{(T_4(m+1) - T_4(m-2))(T_4(n+1)-T_4(n-2))}{25}\\
    &+ f(m+1)f(n+1)\\
    &+ f(m)f(n)\\
    &+ 4(0) + T_4(m-1)T_4(n-1).
 \end{align*}
\end{proof}

\begin{proposition}
The number of ways to tile $L_{4,n}$ with dominoes is given by the explicit formula
\begin{align*}
    &\frac{\frac{6}{b}+5-\frac{1}{b^3}}{(1-\frac{1}{b^2})(1-\frac{b}{a})(1-ab)}b^n-\frac{6b+5-b^3}{(1-\frac{1}{b^2})(1-\frac{b}{a})(1-ab)}(\frac{1}{b})^n\\
    &+\frac{\frac{6}{a}+5-\frac{1}{a^3}}{(1-\frac{1}{a^2})(1-\frac{a}{b})(1-ab)}a^n-\frac{6a+5-a^3}{(1-\frac{1}{a^2})(1-\frac{a}{b})(1-ab)}(\frac{1}{a})^n,
\end{align*}\\
where $a=\frac{1+\sqrt{29}+\sqrt{14+2\sqrt{29}}}{4}$ and $b=\frac{1-\sqrt{29}-\sqrt{14-2\sqrt{29}}}{4}$.
\end{proposition}
\begin{proof}
Use the same method as for the proof of Proposition \ref{Prop:T3explicit}.
\end{proof}

We conjecture that for every positive integer $k$, there exists a real number $x \in \RR$ such that the solutions to the equation $\frac{x^5+1}{x^2(x+1)}=\frac{y^5+1}{y^2(y+1)}$, which are at most two pairs of reciprocal numbers, denoted $y_1, \frac{1}{y_1}, y_2, \frac{1}{y_2}$ allow us to define functions $f_{k,1}(y_1,y_2), f_{k,2}(y_1,y_2), f_{k,3}(y_1,y_2), f_{k,4}(y_1,y_2)$ such that $T_k(n)=f_{k,1}(y_1,y_2) y_1^n+f_{k,2}(y_1,y_2) \frac{1}{y_1}^n+f_{k,3}(y_1,y_2)y_2^n+f_{k,4}(y_1,y_2)\frac{1}{y_2}^n$, where $T_k(n)$ is an explicit formula for the number of ways to tile $L_{k,n}$.

\bibliographystyle{amsalpha}
\bibliography{Dirac}

\begin{bibdiv}
\begin{biblist}

\bib{Cimasoni07}{article}{
      author={Cimasoni, David},
      author={Reshetikhin, Nicolai},
       title={Dimers on surface graphs and spin structures. i},
        date={2007},
        ISSN={1432-0916},
     journal={Communications in Mathematical Physics},
      volume={275},
      number={1},
       pages={187–208},
         url={http://dx.doi.org/10.1007/s00220-007-0302-7},
}

\bib{Contreras20}{article}{
      author={Contreras, Ivan},
      author={Toriyama, Michael},
      author={Yu, Chengzheng},
       title={Gluing of graph laplacians and their spectra},
        date={2020},
     journal={Linear and Multilinear Algebra},
      volume={68},
      number={4},
       pages={710\ndash 749},
         url={https://doi.org/10.1080/03081087.2018.1516727},
}

\bib{Contreras19}{article}{
      author={Contreras, Ivan},
      author={Xu, Boyan},
       title={The graph laplacian and morse inequalities},
        date={2019},
        ISSN={0030-8730},
     journal={Pacific Journal of Mathematics},
      volume={300},
      number={2},
       pages={331–345},
         url={http://dx.doi.org/10.2140/pjm.2019.300.331},
}

\bib{Kasteleyn}{incollection}{
      author={Kasteleyn, Pieter},
       title={Graph theory and crystal physics},
        date={(1967)},
   booktitle={Graph theory and theoretical physics},
      editor={editor, The},
   publisher={Academic Press, London},
       pages={43\ndash 110},
}

\bib{Kenyon2002}{article}{
      author={Kenyon, Richard},
       title={The laplacian and dirac operators on critical planar graphs},
        date={2002},
        ISSN={1432-1297},
     journal={Inventiones mathematicae},
      volume={150},
      number={2},
       pages={409\ndash 439},
         url={https://doi.org/10.1007/s00222-002-0249-4},
}

\bib{KenyonNotes}{article}{
      author={Kenyon, Richard},
       title={{An introduction to the dimer model}},
        date={2003},
     journal={arXiv},
       pages={math/0310326},
}

\bib{Khovanova}{article}{
      author={Khovanova, Tanya},
       title={{Clifford algebras and graphs}},
        date={2008},
     journal={arXiv},
       pages={math/0810.3322},
}

\bib{Knill}{article}{
      author={Knill, Oliver},
       title={{The Dirac operator of a graph}},
        date={2013},
     journal={arXiv},
       pages={math/1306.2166},
}

\bib{Mnev16}{article}{
      author={Mnev, Pavel},
       title={Graph quantum mechanics},
        date={2016},
     journal={Contribution to the 2016 MPIM Jahrbuch},
}

\bib{nica_SpectralGraph}{book}{
      author={Nica, Bogdan},
       title={A brief introduction to spectral graph theory},
   publisher={European Mathematical Society},
     address={Zürich},
        date={(2018)},
        ISBN={9783037196885 3037196882 9783037191880 3037191880},
  url={https://www.amazon.com/Introduction-Spectral-Theory-Textbooks-Mathematics/dp/3037191880},
}

\bib{renaud:hal-03015551}{book}{
      author={Renaud, Pierre},
       title={{Clifford Algebras Lecture Notes on Applications in Physics}},
        date={(2020)},
         url={https://hal.archives-ouvertes.fr/hal-03015551},
}

\bib{Reshetikhin:2014jaa}{article}{
      author={Reshetikhin, Nicolai},
      author={Vertman, Boris},
       title={{Combinatorial Quantum Field Theory and Gluing Formula for
  Determinants}},
        date={2015},
     journal={Lett. Math. Phys.},
      volume={105},
      number={3},
       pages={309\ndash 340},
}

\bib{Yu17}{article}{
      author={Yu, Chengzheng},
       title={Super-walk formulae for even and odd laplacians in finite
  graphs},
        date={2017},
     journal={Rose-Hulman Undergraduate Mathematics Journal},
      volume={18},
         url={https://scholar.rose-hulman.edu/rhumj/vol18/iss1/16},
}

\end{biblist}
\end{bibdiv}

\end{document}